\documentclass[onefignum,onetabnum]{siamart250211} % ,hidelinks,, review

\usepackage[english]{babel}
\usepackage{framed}
\usepackage[normalem]{ulem}
\usepackage{amsmath}

\usepackage{amssymb}
\usepackage{amsfonts}
\usepackage{enumerate}
\usepackage[T1]{fontenc}
\usepackage[utf8]{inputenc}

\usepackage{blkarray}

\usepackage{tcolorbox}

\usepackage{tikz-cd}
\usepackage{mathtools}
\usepackage{csquotes}

\newsiamremark{example}{Example}
\newsiamremark{remark}{Remark}

\newsiamthm{prop}{Proposition}
\newsiamremark{procedure}{Procedure}

\DeclareMathOperator{\rank}{rank}

\newcommand\scalemath[2]{\scalebox{#1}{\mbox{\ensuremath{\displaystyle #2}}}}

\usepackage{marginnote}

\newcommand{\kk}{\Bbbk}
\newcommand{\TT}{\mathbb{T}}
\newcommand{\ZZ}{\mathbb{Z}}

\usepackage{hyperref}

\ifpdf
\hypersetup{
  pdftitle={Nondimensionalization is more science than art},
  pdfauthor={R. Tanburn, S. U. Kumar, D. Hedron, P. Maini, S. Amethyst, E. Dufresne, H. A. Harrington}
}
\fi

\headers{Nondimensionalization is more science than art}{R. Tanburn, et al.}

\title{Nondimensionalization is more science than art%
\thanks{Submitted to the editors November 25, 2025.
\funding{E.D. and H.A.H are grateful to the John Fell Fund that initiated this research. H.A.H. gratefully acknowledges funding from the
Royal Society RGF$\backslash$EA$\backslash$201074 and UF150238 and the EPSRC supporting the center to center collaboration grant EP$\backslash$Z531224$\backslash$1 and Erlangen Programme for AI EP$\backslash$Y028872$\backslash$1.}}
}

\author{
Richard Tanburn%
    \thanks{Mathematical Institute, University of Oxford, Oxford, UK}  %1
\and Danny Hendron%
      \thanks{Department of Mathematics, University of York, York, UK  (\email{danny.hendron@york.ac.uk}, \email{emilie.dufresne@york.ac.uk}}) %3
\and Philip Maini%
      \footnotemark[2] %4
\and Silviana Amethyst%
    \thanks{Max Planck Institute of Molecular Cell Biology and Genetics, Center for Systems Biology Dresden, and Dresden University of Technology, Dresden, DE
  (\email{amethyst@mpi-cbg.de}, \email{harrington@mpi-cbg.de})} % 5
\and Emilie Dufresne%
    \footnotemark[3]  \thanks{Equal contribution.}% these [#] refer to the author number, not the thanks number.
\and Heather A. Harrington%
    \footnotemark[2] \footnotemark[4] \footnotemark[5]
  }

\date{\today}

\begin{document}

\maketitle

\begin{abstract}

When faced with a mathematical model, often the first step is to reduce the complexity of the model by turning variables and parameters into dimensionless quantities.  This process is often performed by hand, relying on a skill practiced over many years, and attempted for small models. Nondimensionalization is often considered an art, as there is no formal method accessible to applied scientists. Here we show how to systematically perform nondimensionalization for arbitrarily sized models described by rational first order ordinary differential equations. We translate and extend an existing approach for computing rational invariants of the maximal scaling symmetry, which combines ideas from differential algebra, invariant theory and linear algebra, to the setting arising in biological models. The modeler inputs the system of equations and our implemented algorithm outputs the nondimensional quantities for the corresponding nondimensionalized model. We extend the algorithm to include initial conditions, and the modeler's choice of invariants, thereby including a larger class of nondimensionalizations. 
We further prove that 
any dimensionally consistent change of variables preserves the dimension of the maximal scaling symmetry. 
We showcase the framework on various models, including the classical Michaelis-Menten equations, which serves as a benchmark for asking and answering specific modeling questions.  
% \sidecomment{This version: \today}
\end{abstract}

% REQUIRED

% REQUIRED
\begin{MSCcodes}
%\href{https://cran.r-project.org/web/classifications/MSC.html}
13A50, 34C14, 34C20, 37N25, 00A71
% 13A50: Actions of groups on commutative rings; invariant theory
% 34C14: Symmetries, invariants
% 34C20: Transformation and reduction of equations and systems, normal forms
% 37N25 — Dynamical systems in biology, chemistry, medicine
% 00A71 — Theory of mathematical modeling
\end{MSCcodes}

\section{Introduction} \label{sec:introduction}

An introductory course in mathematical modeling of physical and biological systems starts with first-order ordinary differential equations (ODE's), which involve variables that vary over time and parameters that represent rates constants. The starting point for analyzing such models is nondimensionalization, which is the process of removing an arbitrary choice of units. Nondimensionalization may reduce the number of parameters in the system 
e.g., reducing the original parameter space $\left(k_1, k_2\right) \in \mathbb{R}^2$ by a dimension by considering a rational function of the parameters $\frac{k_1}{k_2} \in \mathbb{R}$.  
Indeed, the Buckingham-$\pi$ theorem states that if we have $n$ physical variables and parameters, each expressed in some combination of $m$ fundamental independent units (e.g., meters, seconds, moles, liter), then it is possible to express the system in terms of $n-m$ nondimensional quantities \cite{buckingham1914physically, White2008}.  The classical nondimensionalization approach is limited to the reduction of physical units \cite{howison2005practical}. As Segel and Slemrod highlighted in their seminal SIAM Review article \cite{Segel1989}, finding the fundamental dimensions and ``right'' choice of dimensionless quantities is often computed by hand and has been described as an art.

A simple and classic motivating example is enzyme substrate reactions.  Two chemical species $E$ and $S$ reversibly combine to form $ES$, which decomposes irreversibly into $E$ and $P$ \cite{murray}.
\[
E+S\leftrightarrow ES \rightarrow E+P
\]
Incorporating conservation relations, the reactions translate into a parameterized ODE model: 
\begin{align} 
\frac{ds}{dt} = -k_1(e_0 - c)s + k_{-1}c  \qquad \qquad
\frac{dc}{dt} = k_1(e_0 - c)s - k_{-1} c - k_2 c, \label{eq:motivational_example_start_ODE}
\end{align}
where $e_0$ is the initial free enzyme concentration.  The two fundamental physical units in the system are concentration and time:  mol/L and seconds.  The model has more parameters than necessary, and one would like to do analysis on a simpler, nondimensionalized model.

Nondimensionalization is an algebraic process.  For example, one can scale time so that it is measured in seconds or years while not fundamentally changing the system; therefore scaling translates to algebraic torus actions ~\cite{Wynn2014}.
Expressing the units of variables and parameters as combinations of the fundamental independent units (e.g., concentration in moles per liter) induces a scaling symmetry on the variables and parameters. 
Each physical fundamental unit corresponds to an independent one-dimensional scaling symmetry.  The goal of both classical dimensional analysis and the Buckingham-$\pi$ theorem is to remove these symmetries.
Scaling symmetries algebraically translate to rational invariants, and are a special case of a Lie symmetry, which have been used to help solve (higher order) ODE's and PDEs \cite{Hydon2000}. The Lie symmetry approach can in principle be used to find all scaling symmetries and perform nondimensionalization \cite{HubertSedoglavic,Lemaire2012,Sedoglavic2006}; however, these methods are neither easily accessible to applied scientists nor algorithmic \cite{Gilmore2008}.

Hubert and Labahn \cite{Hubert2013c} proposed a purely algebraic and computationally efficient nondimensionalization algorithm relying on linear algebra over the integers.  This powerful and scalable scheme computes the maximal number of scaling symmetries that a system of ODE's can admit.  Na\"ively applying their method to \eqref{eq:motivational_example_start_ODE} yields
\begin{align*}
\frac{du}{d\tau} &= c_{0} v + u v - u  & 
\frac{dv}{d\tau} &= - c_{0} v - c_{1} v - u v + u.
\end{align*}
That is, we do \emph{not} recover the well-known system from \cite{murray}. For that, we must take into account initial conditions and desired invariants or known model constraints.  In this paper, we extend their method, and we can indeed reduce  \eqref{eq:motivational_example_start_ODE} to the familiar dimensionless system from \cite[(6.13)]{murray}:
\begin{align*}
\frac{du}{d\tau} & = -u + (u+K-\lambda)v & u(0)&=1 \\
\epsilon \frac{dv}{d\tau} &= u - (u+K)v    &  v(0)&=0.
\end{align*}
Taking into account scaling on the time variable and several additional constraints, one can obtain in the same framework the familiar reduced system for the slow timescale\footnote{
The details of these examples appear in Section~\ref{sec:applications}.
}:
\begin{align*}
\frac{du}{dT} & = -(1+\sigma)u + \sigma uv + \frac{\rho}{1+\rho}v 
&u(0) &= 1 \\
\epsilon \frac{dv}{dT} & = (1+\sigma)u - \sigma uv - v 
&v(0) &= 0. 
\end{align*}

Independently, Meinsma \cite{meinsma2019dimensional}  notes that 
we might find more scaling symmetries than there are physical fundamental units, but does not offer a practical algorithm. 
The additional scaling symmetries could be understood as structural fundamental units, a refinement of physical fundamental units. Thus, a unified theory of scaling symmetries that respects physical fundamental units yet permits the greatest possible symmetry reduction is needed.

Here we unify classical nondimensionalization and the algebraic algorithmic approach, considering both physical and structural units, to create a theoretical yet practical model reduction procedure that remains faithful to applications. We prove that when we require changes of variables to be dimensionally consistent (i.e., only quantities with the same units can be added together), then the maximal number of scaling symmetries is the same whether the method is applied to the reduced system or the original.
Further, we broaden the reduction algorithm to incorporate practical extensions, including initial conditions, conserved quantities, known parameter combinations, and nondimensional quantities. 
We showcase this framework using the \emph{Michaelis-Menten kinetics} \cite{Gunawardena2014,MichaelisMenten} chemical reaction network, recovering the standard nondimensionalization as well as the alternative Segel-Slemrod formulation.  We also include an example with a reduction beyond the number of fundamental units in Section~\ref{sec:cell_cycle_control}.
This algorithmic approach is implemented in the Python software library \verb|desr| (Differential Equation Symmetry Reduction) \cite{desr}.

Section~\ref{sec:preliminaries} lays out the mathematical preliminaries of the tools and theory. 
We start with a description of dimensional analysis and scaling symmetries (Section~\ref{sec:dimensional_analysis}). Next we define a scaling symmetry associated to an integer matrix (Section~\ref{sec:matrix_scaling_actions}), and describe the main computational tools: Hermite normal form (\ref{sec:hermite_normal_form}), Hermite multipliers (\ref{sec:hermite_multiplier}) and how to construct the exponent matrix (Section~\ref{sec:exponent_matrix}). All of these preliminaries are tied together in Section~\ref{sec:maximal_scaling_symmetry} to describe the method for computing the maximal scaling action of a system of ODE's.

Section~\ref{sec:eliminating_scaling_symmetries} contains the mathematical pipeline, theoretical and algorithmic contributions.   We begin by covering some known results.  Section~\ref{sec:reduced_systems} describes the algorithmic reduction of a system, including the independent variable, using Theorem~\ref{thm:hubert_labahn_6.5} via Hermite multipliers on the exponent matrix.  Section~\ref{sec:nondimensionalization} describes how to use explicitly find the invariants of a system under the maximal scaling action, and Theorem~\ref{thm:hubert_labahn_7.1} describes the substitutions for reduction. Further, the method generalizes to broader forms of Hermite multipliers. In Section~\ref{sec:practical_considerations}, we improve upon the previous reduction method to include initial conditions or known conserved quantities or constants that the modeler may wish to incorporate.  Prescribed parameter combinations may not be dimensionally compatible with a model, so Theorem~\ref{thm:when_can_rewrite_using_parameter_combinations} in Section~\ref{sec:extending_choice_invariants} supports extending a compatible choice of invariant parameters to a complete set of invariants.  Finally, in Section~\ref{sec:Change_of_variables} we prove Theorem~\ref{thm:ChangeOfVariables}, establishing that the dimension of the model reduction is the same regardless of the change of variables, as long as that change is dimensionally consistent.  

Section~\ref{sec:applications} showcases the nondimensionalisation pipeline on mathematical models. We take as a case study, Michaelis-Menten kinetics:  Section~\ref{sec:example_michaelis_menten} illustrates the basic method; Section~\ref{sec:michaelis_menten_comparison_with_classical} covers the incorporation of initial conditions and the classical Michaelis constant $K_m$; Section~\ref{sec:application_segel_slemrod} uses the reduction method in conjunction with multiple time scales; Sections~\ref{sec:application_new_nondimen} and \ref{sec:michaelis_menten_four_variables} show that one can obtain new nondimensionalizations of familiar systems.  Section~\ref{sec:cell_cycle_control} presents this approach works on larger system of cell cycle control. Finally, Section~\ref{sec:compartment model} illustrates the impact of Theorem~\ref{thm:ChangeOfVariables} by way of a linear compartment vaccine injection model.

We conclude in Section~\ref{sec:conclusion}.

\section{Preliminaries} \label{sec:preliminaries}

We begin by bringing together relevant background material in order to provide the tools for algorithmically nondimensionalizing a dynamical system model. This is done by seeing how to compute the \emph{scaling symmetries} of a system of differential equations, and using this information to then produce a reduced system. The algorithm itself relies on linear algebra over the integers, specifically, the \emph{Hermite normal form} of an integer matrix.

\subsection{Dimensional analysis  and scaling symmetries} \label{sec:dimensional_analysis}

The quantities in a model have units, and all of these units can be written as a product of powers of independent fundamental units. For example, the speed of an object in meters per second has units \(m/s=ms^{-1}\) and the area in square meters has units \(m^2\), both expressed in terms of the fundamental units of time \(s\) and distance \(m\). There usually is a ``natural'' choice of fundamental units depending on the situation described, but from a mathematical point of view they are all equivalent (like all the choices of basis of a vector space). For example, it is natural to describe the speed of an object in terms of units of distance and time, but from a mathematical point of view, it would be equally valid to take the units of speed and time as fundamental, and describe distance in terms of speed and time. Notably, in the International System of Units (SI) all units are expressed in terms of some universal constants which have composite units.  The speed of light in a vacuum is such a constant, and a meter (the base unit of length in SI) is defined to be the distance traveled by light in a vacuum during a time interval of \(1/299 792 458\) of a second \cite{newell2019international}.

One key observation (already present in \cite{Wynn2014}) is that we can ``change the units'' of fundamental units independently. For example, in a system involving time and distance we can change from seconds to hours as our unit of time {\em without} having to change the unit of distance. This corresponds to scaling: $1$ second is $1/60$ minutes, one day is $24$ hours, one second is $1000$ milliseconds, and so on. Under certain assumptions (the universe is infinite, matter is continuous, \ldots), we can scale by any positive real number to obtain new units. Moreover, when we rescale a fundamental unit, composite (or derived) units involving this fundamental unit are also scaled according to the relationship. 
For example, scaling from seconds to minutes and then hours, a speed of \(1\) meter per second corresponds to \(60\) meters per minute and \(3600\) meters per hour. 

Thus, if there are $k$ fundamental units, then we have $k$ independent actions of the multiplicative group of real positive numbers on the quantities involved in the model, or equivalently, an action of the $k$-dimensional real positive multiplicative group. 

Dimensionless quantities are ``composite quantities'' which are not affected by rescaling the fundamental units. By composite quantities, we mean rational functions of the quantities we started with. In other words, dimensionless quantities are the composite quantities invariant under the action of the $k$-dimensional real positive multiplicative group induced by the fundamental units. In particular, since all units can be described in terms of $7$ constants \cite{newell2019international}, this means that in the International System of Units (SI) there is a $7$-dimensional real positive multiplicative group acting on our physical universe.

The Buckingham-\(\pi\) Theorem famously says that if there are $n$ quantities and $k$ fundamental units, then all dimensionless quantities can be expressed in terms of $n-k$ independent dimensionless quantities \cite{buckingham1914physically}. When we think of fundamental units as corresponding to independent scaling actions, this result is a consequence of the basic result about the rational invariants of a linear action by an algebraic torus.%
\footnote{In Algebraic Geometry, the multiplicative group of a field is known at the \emph{algebraic torus}. This name is used because within the theory of algebraic groups, these groups play a similar role as the topological torus (a product of circle groups) plays in the theory of compact Lie Groups (see paragraph 8.5 in \cite{Borel1991}).}
Indeed, it is well known that the rational invariants of a linear action of an $k$-dimensional algebraic torus on a $n$-dimensional space are rational of transcendence degree $n-k$, meaning that all rational invariants can be written in terms of $n-k$ field generators (see for example \cite{VinbergPopov1989}).  As field generators are far from unique, so are ``basic dimensionless'' quantities.

In the following sections we introduce precise notation to describe and justify the method of Hubert and Labahn. Instead of using the scaling symmetries corresponding to the fundamental units present in the physical system, which we call \emph{physical fundamental units}, we start with the assumption that each quantity that appears in the equation has its own independent unit, and then use the equations to deduce relations between units in the model. 

Equations must be dimensionally consistent; that is, quantities that are added together to have the same units. The first step of the method is to find the maximal scaling symmetries consistent with the ODE's. These scaling symmetries always include the scaling symmetries corresponding to the physical fundamental units, but may also include more; Following Meinsma, we refer to these maximal scaling symmetries as \emph{structural fundamental units} \cite{meinsma2019dimensional}.
Once we have structural fundamental units, Hubert-Labahn's method produces a reduced system which only consists of dimensionless quantities. In particular, if we find more structural fundamental units than there are physical fundamental units, then we achieve more reduction than we would via traditional dimensional analysis.

\subsection{Matrix scaling actions} \label{sec:matrix_scaling_actions}

Let $z = (z_1(t),z_2(t),\ldots,z_n(t))$ be a vector of real-valued variables. Define
\begin{equation*}
    \frac{dz}{dt} \coloneqq \left(\frac{dz_1}{dt},\frac{dz_2}{dt},\ldots,\frac{dz_n}{dt} \right), \quad z^{-1} \coloneqq (z_1^{-1},z_2^{-1},\ldots,z_n^{-1}).
\end{equation*}

Let $*$ be element-wise multiplication with
\begin{equation*}
    (x_1,\ldots,x_n) * (y_1,\ldots,y_n) \coloneqq (x_1 y_1,\ldots,x_n y_n).\label{eqn:multiply_two_vectors}
\end{equation*}

For a vector of non-zero real numbers $x$ and matrix of integers $A \in M_{m\times n}(\mathbb{Z})$, define
\[
    x^A := (x_1^{A_{1,1}}x_2^{A_{2,1}}\cdots x_m^{A_{m,1}},\ldots,x_1^{A_{1,n}}x_2^{A_{2,n}} \cdots x_m^{A_{m,n}}).\]

The \emph{m-dimensional real algebraic torus} is the set \(\mathbb{T}^m\) of all vectors of non-zero real numbers of the form \(\lambda = (\lambda_1,\ldots,\lambda_m)\). It is a group under element-wise multiplication as defined above.

Consider a system of differential equations
\begin{equation}
    \frac{dz}{dt} = f(t,z) = (f_1(t,z),\ldots,f_n(t,z)).  \label{eqn:system_diff_eqns}
\end{equation}

\begin{definition}  \label{defn:scaling_symmetry}
    We say that $A \in M_{m \times n}(\mathbb{Z})$ induces a \emph{scaling symmetry} of a system \eqref{eqn:system_diff_eqns} if for any solution $z$ and any $\lambda \in \TT^m$ 
    \begin{equation*}
    \lambda^A * z= (\lambda_1^{A_{1,1}}\lambda_2^{A_{2,1}}\cdots \lambda_r^{A_{m,1}}z_1,\ldots,\lambda_1^{A_{1,n}}\lambda_2^{A_{2,n}} \cdots \lambda_r^{A_{m,n}}z_n)
    \end{equation*}
    is also a solution.
    
\end{definition}

\begin{example}
    Consider the system of differential equations
\begin{equation} \label{eqn:matrix_scaling_ex1}
    \frac{dz_1}{dt} = tz_1\left(1 - \frac{z_1^2}{z_2} \right), \quad \frac{dz_2}{dt} = tz_2\left(1 - \frac{z_2}{z_1^2} \right).
\end{equation}
   The matrix $A = (1 \, 2) \in M_{1 \times 2}(\mathbb{Z})$ induces a scaling symmetry since for all \(\lambda_1\in\TT\), $(\nu_1, \nu_2) = (\lambda_1)^A * (z_1,z_2) = (\lambda_1,\lambda_1^2) * (z_1,z_2) = (\lambda_1 z_1,\lambda_1^2 z_2)$ also satisfies \eqref{eqn:matrix_scaling_ex1}.
\end{example}

\subsection{Hermite normal form} \label{sec:hermite_normal_form}

In linear algebra, square diagonal matrices are particularly nice to work with: entries are precisely their eigenvalues.  One cannot always express a matrix in diagonal form; the `next best thing' is the \emph{Jordan normal form}, which transforms a square matrix into one which is `almost' diagonal.  
However, integer matrices cannot always be put into Jordan normal form using row (or column) operations.  This paper extensively uses integer linear algebra, so instead we use \emph{Hermite normal form}, which fortunately can be computed in polynomial time \cite{HMM}.

\begin{definition} \label{defn:row_hermite_normal_form} 
\cite[Section 2.4.2]{cohen}.
Let $B \in M_{m \times n}(\mathbb{Z})$. We say that $B$ is in \emph{row Hermite normal form} if it equals some matrix $H \in M_{m \times n}(\mathbb{Z})$ such that:
\begin{enumerate}
   \item $H$ is upper-triangular ($h_{ij} = 0$ for $i > j$);
   \item All zero-rows are located beneath the non-zero rows of $H$;
   \item The first non-zero entry of each row (called the \emph{pivot}) is positive, and is strictly to the right of the pivot of the row above;
   \item All elements below a pivot are zero, and all elements above a pivot are strictly smaller than it.
\end{enumerate}
\end{definition}

\begin{example} 
The following matrix is in row Hermite normal form: 
$$
    \begin{pmatrix}
    1 & 0 & 50 & -8 & 4 \\
    0 & 3 & 12 & -2 & 4 \\
    0 & 0 & 54 & -11 & 0 \\
    0 & 0 & 0 & 0 & 0
    \end{pmatrix}.
$$
\end{example}

There is an analogous definition of the \emph{column Hermite normal form}, which requires $H$ to be lower-triangular with zero-columns on the right. A pivot element is then the first non-zero entry of a column, which must lie beneath the previous pivot element, with zero-entries to the right and strictly smaller entries to the left. For example, $H = \bigl( \begin{smallmatrix}2 & 0 & 0\\ 3 & 5 & 0\\\end{smallmatrix}\bigr)$ is in column Hermite normal form. 

\subsection{Hermite multipliers} \label{sec:hermite_multiplier}

Given an arbitrary matrix $B \in M_{m \times n}(\mathbb{Z})$, we can transform $B$ into column Hermite normal form using Gaussian elimination as in standard linear algebra. Namely, we may reorder columns, multiply a column by $-1$, and/or add an integer multiple of one column to another.
This process turns out to be equivalent to finding a unimodular matrix%
\footnote{A \emph{unimodular} matrix is a square integer matrix with determinant $\pm1$. Unimodular is the integer analogue of invertible in classical linear algebra.} 
$V \in M_n(\mathbb{Z})$ such that $BV$ is in column Hermite normal form \cite[Theorem 2.4.3]{cohen}. $V$ is called the \emph{Hermite multiplier} of $B$. 
  
  If $B$ has rank\footnotemark $r$ and $m=r<n$, then we can partition $V$ further into
\begin{equation} \label{eqn:hermite_multiplier_decomp}
    V = \begin{bmatrix} V_{\mathfrak{a}} & V_{\mathfrak{b}} \end{bmatrix}, \quad \quad V^{-1} = W = \begin{bmatrix} W_{\mathfrak{a}} \\ W_{\mathfrak{b}} \end{bmatrix},
\end{equation}
\footnotetext{
As in classical linear algebra, for matrices of integers the column and row rank coincide, so the term \emph{rank} refers to either of these notions. Further, the rank of $B \in M_{m \times n}(\mathbb{Z})$ is the dimension of the vector space over the real (or complex) numbers spanned by its rows/columns.}
where $V_{\mathfrak{a}}$ is the first $r$ columns of $V$, and $W_{\mathfrak{a}}$ the first $r$ rows of $W$. 

There is an analogous definition of a Hermite multiplier $U \in M_r(\mathbb{Z})$ of $B$ such that $UB$ is in \emph{row} Hermite normal form.  Similarly to \eqref{eqn:hermite_multiplier_decomp}, when $B$ has full rank we can partition $U$ into
\begin{equation} \label{eqn:hermite_multiplier_decomp_square}
    U = \begin{bmatrix} U_{\mathfrak{a}} \\ U_{\mathfrak{b}} \end{bmatrix}, \quad \quad U^{-1} = X = \begin{bmatrix} X_{\mathfrak{a}} & X_{\mathfrak{b}} \end{bmatrix}.
\end{equation}

Since we use both kinds of Hermite multiplier, we write $V$ and $W$ when referring to column Hermite normal form, and $U$ and $X$ for row Hermite formal form. 

\begin{example} 
Take $B = \bigl( \begin{smallmatrix}3 & -1 & 0\\ 4 & 0 & 2\\\end{smallmatrix}\bigr)$. Using Gaussian elimination (not necessarily in a unique way) to transform $B$ into \emph{column} Hermite normal form:
\begin{equation*}
    B = \left( \begin{smallmatrix}3 & -1 & 0\\ 4 & 0 & 2\\\end{smallmatrix}\right) \rightsquigarrow \left( \begin{smallmatrix}0 & -1 & 0\\ 4 & 0 & 2\\\end{smallmatrix}\right) \rightsquigarrow \left( \begin{smallmatrix}0 & -1 & 0\\ 0 & 0 & 2\\\end{smallmatrix}\right) \rightsquigarrow \left( \begin{smallmatrix}-1 & 0 & 0\\ 0 & 0 & 2\\\end{smallmatrix}\right) \rightsquigarrow \left( \begin{smallmatrix}-1 & 0 & 0\\ 0 & 2 & 0\\\end{smallmatrix}\right) \rightsquigarrow \left( \begin{smallmatrix}1 & 0 & 0\\ 0 & 2 & 0\\\end{smallmatrix}\right) = H
\end{equation*}

Representing each transformation as a matrix, multiplied together they give a column Hermite multiplier $V$:
\begin{equation*}
    V = \left( \begin{smallmatrix}1 & 0 & 0\\ 0 & 1 & 0 \\ -2 & 0 & 1\\\end{smallmatrix}\right) \cdot \left( \begin{smallmatrix}1 & 0 & 0\\ 3 & 1 & 0 \\ 0 & 0 & 1\\\end{smallmatrix}\right) \cdot \left( \begin{smallmatrix}0 & 1 & 0\\ 1 & 0 & 0 \\ 0 & 0 & 1\\\end{smallmatrix}\right) \cdot \left( \begin{smallmatrix}-1 & 0 & 0\\ 0 & 1 & 0 \\ 0 & 0 & 1\\\end{smallmatrix}\right) \cdot \left( \begin{smallmatrix}1 & 0 & 0\\ 0 & 0 & 1 \\ 0 & 1 & 0\\\end{smallmatrix}\right) = \left( \begin{smallmatrix}0 & 0 & 1\\ -1 & 0 & 3 \\ 0 & 1 & -2\\\end{smallmatrix}\right).
\end{equation*}

\noindent
One can readily verify that $BV = H$.
\end{example}

The Hermite multiplier need not be unique. To remove this ambiguity, we introduce the \emph{normal Hermite multiplier}:
\begin{definition} \cite[Section 2.2]{Hubert2013c}. 
For $B \in M_{r \times n}(\mathbb{Z})$ of full row rank, let $V \in M_n(\mathbb{Z})$ be a Hermite multiplier of $B$, so that $H = BV$ is in column Hermite normal form. Then $V$ is said to be the \emph{normal Hermite multiplier} of $B$ if:
\begin{enumerate}
    \item $V = \begin{bmatrix} V_{\mathfrak{a}} & V_{\mathfrak{b}} \end{bmatrix}$ is such that $V_{\mathfrak{b}}$ is in column Hermite normal form.
    \item $V_{\mathfrak{a}}$ is reduced with respect to the pivot rows\footnotemark of $V_{\mathfrak{b}}$.
\end{enumerate}
\end{definition}

\footnotetext{This means that if $i_1<i_2<\ldots <i_{n-r}$ are the pivot rows for $V_{\mathfrak{b}}$ then for each $1 \leqslant j \leqslant n-r$, we have that $0 \leqslant (V_{\mathfrak{a}})_{i_j,k} < (V_{\mathfrak{b}})_{i_j,j}$ for each $1 \leqslant k \leqslant r$.}

\begin{theorem} \label{thm:normal_hermite_mult}
Two important properties of Hermite multipliers:
\begin{enumerate}
\item \emph{The normal Hermite multiplier exists and is unique} \cite[Proposition 2.3]{Hubert2013c}.

\item\emph{If $V$ is a Hermite multiplier of $B$, then the last $n-r$ columns of $V$ form a basis for the kernel of $B$} \cite[Proposition 2.4.9]{cohen}.
\end{enumerate}
\end{theorem}

\subsection{Exponent matrix} \label{sec:exponent_matrix}

\begin{definition} \label{def:exponent_vector_matrix}
Let $f(z_1,\ldots,z_n)$ be a rational function, so $f = \frac{p}{q}$ for coprime\footnotemark \footnotetext{Two polynomials $p,q$ are \emph{coprime} if they have no common factors.} polynomials $p,q$. For each monomial $z_1^{a_1} z_2^{a_2} \ldots z_n^{a_n}$ appearing in $p$ and $q$, we define its \emph{monomial exponent vector} to be the vector $(a_1,\ldots,a_n)$ of powers of the variables in the monomial. The \emph{exponent matrix} $K_f$ of $f$ is computed in these two steps:
\begin{enumerate}
    \item Form a matrix whose columns are the monomial exponent vectors of each term in $p$ and $q$ (the order of the columns is arbitrary).
    \item Pick and remove a column of the matrix, and then subtract it from the remaining columns. The choice of column does not matter, but we usually choose the simplest monomial exponent vector of $q$.
\end{enumerate}
\end{definition}

\begin{example}

With $f = \frac{z_1 z_2}{1}$, so $p = z_1 z_2$ and $q = 1$, the monomial exponent vectors of $p$ and $q$ are $(1,1)$ and $(0,0)$ respectively. We pick the latter for subtraction and so we obtain $K_f = \bigl( \begin{smallmatrix}1\\ 1\\\end{smallmatrix}\bigr)$.
    
On the other hand, if $g = \frac{z_1 z_2 + z_3}{z_2 z_3}$, then $p = z_1 z_2 + z_3$ and $q = z_2 z_3$. This yields monomial exponent vectors $(1,1,0)$, $(0,0,1)$, and $(0,1,1)$. Subtracting and removing the last vector gives $K_g = \left( \begin{smallmatrix}1 & 0\\ 0 & -1 \\ -1 & 0\\\end{smallmatrix}\right)$. The columns of $K_g$  correspond to $\frac{z_1}{z_3}$ and $\frac{1}{z_2}$,  so removing from $K_g$ the column which corresponds to the monomial exponent vector $z_2 z_3$ corresponds to dividing both the numerator and denominator of $g = \frac{p}{q}$ by $z_2 z_3$.
\end{example}

\subsection{Maximal scaling action} \label{sec:maximal_scaling_symmetry}

Given the rational functions $f_1,\ldots,f_m$ in variables $z_1,\ldots,z_n$, we seek the maximal scaling action fixing each \(f_i\).

\begin{definition}
        The \emph{dimension} $r$ of a scaling action is the rank of the scaling matrix $A$.
    A scaling action $A$ is \emph{maximal} among scaling actions fixing a given list of rational functions \(\{f_1,\ldots,f_m\}\) if any other scaling action $B$ fixing \(\{f_1,\ldots,f_m\}\) satisfies $\rank(B) \leqslant \rank(A)$. 
\end{definition}

We now have the ingredients to compute the maximal scaling matrix $A$.
\begin{procedure} \label{proc:max_scaling}
For the rational functions $f_1,\ldots,f_m$, compute the dimension $r$ and scaling matrix $A$:
\begin{enumerate}
    \item Concatenate exponent matrices to form $K = [K_{f_1} \mid \ldots \mid K_{f_n}]$. 
    \item Compute the row Hermite normal form of $K$, so that $H = UK$. 
    \item Let $r$ be the number of zero-rows at the bottom of $H$, and $A \in M_{r \times n}(\mathbb{Z})$ be the bottom $r$ rows of the Hermite multiplier $U$. 
\end{enumerate} 
\end{procedure}

It is important to note that $A$ need not be unique; it can be acted upon by elementary row operations to give other equivalent scaling actions.

\begin{example}
Let $F = (f_1,f_2) = \big(\frac{z_3z_4 + z_1z_2}{z_1z_3}, \frac{z_1z_4 + z_3}{z_2 + z_5} \big)$. Form the exponent matrices for both functions, subtracting and removing the monomial exponent vectors for $z_1z_3$ and $z_5$ respectively from $f_1$ and $f_2$, giving
\begin{equation*}
K = [K_{f_1} \mid K_{f_2}] = 
    \left[\begin{array}{cc|ccc}
    -1 & 0 & 0 & 1 & 0 \\
    0 & 1 & 0 & 0 & 1 \\
    0 & -1 & 1 & 0 & 0 \\
    1 & 0 & 0 & 1 & 0 \\
    0 & 0 & -1 & -1 & -1
    \end{array}\right].     
\end{equation*}

The row Hermite normal form decomposition is
\begin{equation*}
UK =
    \left[\begin{array}{cccccc}
    0 & 1 & 1 & 1 & 1 \\
    0 & 1 & 0 & 0 & 0 \\
    0 & 1 & 1 & 0 & 0 \\
    0 & -1 & -1 & 0 & -1 \\
    \hline
    -1 & -2 & -2 & -1 & -2
    \end{array}\right] \cdot
    K 
    = \left[\begin{array}{cccccc}
    1 & 0 & 0 & 0 & 0 \\
    0 & 1 & 0 & 0 & 1 \\
    0 & 0 & 1 & 0 & 1 \\
    0 & 0 & 0 & 1 & 0 \\
    \hline
    0 & 0 & 0 & 0 & 0
    \end{array}\right] = H.
\end{equation*}
$H$ has one row of zeroes, so $r=1$. The bottom row of $U$ gives the scaling matrix $A = [1 \ 2 \ 2 \ 1 \ 2]$, which we have multiplied by $-1$ for simplicity since elementary row operations give equivalent scaling actions. Thus $A$ is the maximal scaling action, and indeed for any $\lambda \in \mathbb{T}^1 = \mathbb{R}^*$,
\begin{align*}
    f_1(\lambda^A * z) &= \frac{(\lambda z_1)(\lambda^2 z_2) + (\lambda^2 z_3)(\lambda z_4)}{(\lambda z_1)(\lambda^2 z_3)} = \frac{\lambda^3(z_1z_2 + z_3z_4)}{\lambda^3 z_1z_3} = \frac{z_1z_2 + z_3z_4}{z_1z_3} = f_1(z) \\
    f_2(\lambda^A * z) &= \frac{(\lambda z_1)(\lambda z_4) + \lambda^2 z_3}{\lambda^2z_2 + \lambda^2 z_5} = \frac{\lambda^2 (z_1 z_4 + z_3)}{\lambda^2 (z_2 + z_5)} = \frac{z_1 z_4 + z_3}{z_2 + z_5} = f_2(z).
\end{align*}
\end{example}

\begin{remark}  
    In Definition~\ref{def:exponent_vector_matrix}, we made arbitrary choices for both the ordering of the columns in $K_f$ and also which column we subtract and remove. The beginning of \cite[Section 5]{Hubert2013c} explains why different choices do not affect the maximal scaling action of $A$ on the rational function $f$.
\end{remark}

Let us use the results from \cite[Section 5]{Hubert2013c} to prove that Procedure~\ref{proc:max_scaling} computes a maximal scaling action. 
Suppose we are given rational functions $f_1,\ldots,f_m \in k(z_1,\ldots,z_n)$. Each $f_i$ is $\mathbb{T}_A$-invariant if and only if $AK_{f_i} = 0$. Thus $f_1,\ldots,f_m$ are $\mathbb{T}_A$-invariant if and only if $AK = 0$ where $K = [K_{f_1} \mid \ldots \mid K_{f_m}]$. So we need to check that our matrix $A$ satisfies this condition, and then prove the maximality of the scaling action.

\begin{prop} \label{prop:hubert5_1}
\cite[Proposition 5.1]{Hubert2013c}.
\emph{Let $Z \in M_{n \times k}(\mathbb{Z})$. Let $U \in M_n(\mathbb{Z})$ be unimodular such that $UZ=H$ is in row Hermite normal form, with precisely $r$ zero-rows. Let $A$ be the bottom $r$ rows of $U$, so that $A \in M_{r \times n}(\mathbb{Z})$. Then:}

\begin{enumerate}
    \item[(1)] $AZ=0$.
    \item[(2)] $A$ \emph{has column Hermite normal form} $A = [I_r \ \ 0]$.
    \item[(3)] \emph{An integer matrix $B$ satisfies $BZ=0$ if and only if there exists an integer matrix $M$ such that $B=MA$.}
\end{enumerate}
\end{prop}

Using Proposition~\ref{prop:hubert5_1}, setting $Z=K$ gives all information required to show that the matrix $A$ gives a maximal scaling action.   \emph{(1)}~shows that $A$ ensures that each $f_i$ is $\mathbb{T}_A$ invariant. \emph{(2)}~shows that $A$ has rank $r$.  \emph{(3)}~shows that if $B$ is another scaling symmetry then $B=MA$ (so $B$ has rank at most that of $A$), and thus $A$ is maximal.

%------------------------------------------------------%

\section{Theory for eliminating scaling symmetries in dynamical systems} \label{sec:eliminating_scaling_symmetries}
This section uses the tools developed in Section~\ref{sec:preliminaries} to provide a theoretical foundation for the method of finding a nondimensionalization and subsequently producing a dimensionless system of differential equations.
Finding the maximal number of independent scalings of a dynamical system which map any solution to another solution allows us to construct an equivalent reduced system by eliminating these symmetries.  We show how to recover the solutions of the original system from a reduction. In other words, by finding the maximal scaling symmetries of a dynamical system, we produce a new system with fewer symbols, whose solutions correspond to those of the original. Sections \ref{sec:practical_considerations} and \ref{sec:completing_partial_sets_of_invariants} show how we can extend the methods described before to incorporate initial conditions and known invariants. 
In Section \ref{sec:Change_of_variables}, we show that a sensible change of variables (i.e., those where we assume that two quantities added together have the same units) will never give rise to more scaling symmetries than were present in the original system.

\subsection{Construction of reduced systems} \label{sec:reduced_systems}

Definition~\ref{defn:scaling_symmetry} requires that both sides of the differential equation scale in the same way.  For the reduction algorithm, it will be simpler to require instead that one expression is invariant under a scaling action.  Thus, rewrite the differential equations as
\begin{equation} \label{eqn:our_ODE's}
    \frac{dz}{dt} = f(t,z) = \frac{z * F(t,z)}{t},
\end{equation}
\noindent 
with $F_i(t,z) \coloneqq \frac{t f_i(t,z)}{z_i}$.  We will focus on these rational functions instead of the $f_i$, in light of the following lemma. For $\bar{A} = [A_0 \ \ A] \in M_{r \times (n+1)}(\mathbb{Z})$, we say that $\mathbb{T}_{\bar{A}}$ defines a scaling symmetry for \eqref{eqn:our_ODE's} if for any solution $(t,z)$ and $\lambda \in \mathbb{T}^r$, we have that $\lambda^{\bar{A}} * (t,z)$ is also a solution.

\begin{lemma}
\emph{$\mathbb{T}_{\bar{A}}$ defines a scaling symmetry for \eqref{eqn:our_ODE's} if and only if $F_i$ is a rational invariant for each $1 \leqslant i \leqslant n$. In other words, $F_i(\lambda^{\bar{A}} * (t,z)) = F_i(t,z)$.}
\end{lemma}

\begin{proof}
Suppose that $(t,z)$ is a solution to \eqref{eqn:our_ODE's}. Then for each $1 \leqslant i \leqslant n$,
\begin{align*}
    \frac{d(\lambda^{A_i} z_i)}{d(\lambda^{A_0} t)} &= \frac{\lambda^{A_i} z_i}{\lambda^{A_0} t} F_i(\lambda^{\bar{A}} * (t,z)), \quad \textrm{but also} \\
    \frac{d(\lambda^{A_i} z_i)}{d(\lambda^{A_0} t)} &= \lambda^{A_i} \frac{dz_i}{dt} \frac{dt}{d(\lambda^{A_0} t)} = \frac{\lambda^{A_i}}{\lambda^{A_0}} \frac{dz_i}{dt} = \frac{\lambda^{A_i}}{\lambda^{A_0}} \frac{z_i F_i(t,z)}{t}.
\end{align*}
Thus $\lambda^{\bar{A}} * (t,z)$ is a solution to \eqref{eqn:our_ODE's} if and only if every $F_i$ is $\mathbb{T}_{\bar{A}}$-invariant.
\end{proof}

So that the algorithm can also scale time, we introduce a new dependent variable $z_0(t)$, which acts like a scaled version of the independent variable $t$. Consider the following dynamical system:
\begin{equation} \label{eqn:ODE_time_scaled}
    \frac{d \bar{z}}{dt} = \frac{\bar{z}}{t} * \bar{F}(\bar{z}), \quad \textrm{where} \quad \bar{z} = (z_0,z_1,\ldots,z_n) \quad \textrm{and} \quad \bar{F} = (1,F_1,\ldots,F_n).
\end{equation}
Note that the first equation of the dynamical system in \eqref{eqn:ODE_time_scaled} is $\frac{dz_0}{dt} = \frac{z_0}{t}$, which has solution $z_0(t) = ct$ for some constant $c \in \mathbb{R}^*$.
\begin{lemma}
\emph{If $z(t) = (z_1(t),\ldots,z_n(t))$ is a solution to a dynamical system as in \eqref{eqn:our_ODE's}, then $\bar{z}(t) = (t,z_1(t),\ldots,z_n(t))$ is a solution to \eqref{eqn:ODE_time_scaled}. Conversely, if $\bar{z}(t) = (z_0(t),z_1(t),\ldots,z_n(t))$ is a solution of \eqref{eqn:ODE_time_scaled} with $z_0(t) = ct$ for some $c \in \mathbb{R}^*$, then $z(t) = (z_1(\frac{t}{c}),\ldots,z_n(\frac{t}{c}))$ is a solution to \eqref{eqn:our_ODE's}.}
\end{lemma}

\begin{proof}
The first statement is immediate, since $\frac{dt}{dt} = 1$ indeed equals $\frac{z_0}{t}$ because $z_0(t) = t$, and the equations for $z_1,\ldots,z_n$ are unchanged. For the converse, suppose that $\bar{z}(t)$ is a solution to \eqref{eqn:ODE_time_scaled}, and let $Z_i(t) = z_i(\frac{t}{c})$. Thus, for each $1 \leqslant i \leqslant n$:
\begin{equation}
    \frac{dZ_i(t)}{t} = \frac{dz_i(\frac{t}{c})}{dt} = \frac{z_i(\frac{t}{c})}{t} F_i \left(z_0\left(\frac{t}{c}\right),z_1\left(\frac{t}{c}\right),\ldots,z_n\left(\frac{t}{c}\right) \right) = \frac{Z_i(t)}{t} F_i(t,Z(t)).
\end{equation}
So if $z_0(t) = ct$ for some $c \in \mathbb{R}^*$, then $Z(t) = z(\frac{t}{c})$ is a solution to \eqref{eqn:our_ODE's}.
\end{proof}

\begin{theorem} \label{thm:hubert_labahn_6.5}
\cite[Theorem 6.5]{Hubert2013c}.
\emph{Let $\bar{A} = [A_0 \ \ A] \in M_{r \times (n+1)}(\mathbb{Z})$. Suppose that the rational functions $\bar{F}(\bar{z}) = (1,F_1(\bar{z}),\ldots,F_n(\bar{z}))$ are invariant under the $\mathbb{T}_{\bar{A}}$-action on the variables $z_0,z_1,\ldots,z_n$. Let $V = [V_{\mathfrak{a}} \ \ V_{\mathfrak{b}}]$ be a Hermite multiplier of $\bar{A}$, with inverse $W = \left[ \begin{smallmatrix} W_{\mathfrak{a}} \\ W_{\mathfrak{b}} \end{smallmatrix} \right]$.}

\begin{enumerate}
    \item \emph{If $z(t)$ is a solution of \eqref{eqn:our_ODE's} such that no components of $z(t)$ vanish
    % \footnote{Suppose we are interested in solutions in the range $t \in [T_0,T_1]$ where $0 \leqslant T_0 \leqslant T_1 < \infty$. Then each $z_i(t)$ has no roots in the interval $[T_0,T_1]$.}%
     and $\bar{z}(t) = (t,z_1(t),\ldots,z_n(t))$, then $[x(t) \ \ y(t)] = [\bar{z}^{V_{\mathfrak{a}}} \ \ \bar{z}^{V_{\mathfrak{b}}}]$ is a solution to}
    \begin{equation} \label{eqn:ODE's_after_scaling}
        \frac{dy}{dt} = \frac{y}{t} * (\bar{F}(y^{W_{\mathfrak{b}}}) \cdot V_{\mathfrak{b}}), \quad \quad \frac{dx}{dt} = \frac{x}{t} * (\bar{F}(y^{W_{\mathfrak{b}}}) \cdot V_{\mathfrak{a}}).
    \end{equation}
    
    \item \emph{Suppose $y(t),x(t)$ are respectively solutions to \eqref{eqn:ODE's_after_scaling} such that none of their components vanish. If $\bar{z}(t) = (z_0(t),\ldots,z_n(t)) = [x(t) \ \ y(t)]^W$, then $z_0(t) = ct$ for some $c \in \mathbb{R}^*$ and $z(t) = (z_1(\frac{t}{c}),\ldots,z_n(\frac{t}{c}))$ is a solution to the original dynamical system \eqref{eqn:our_ODE's}.}
\end{enumerate}

\end{theorem}

\begin{remark} \label{remark:hubert_labahn6.5_rational_invariants_by_substitution}
From a system of $n$ differential equations, Theorem~\ref{thm:hubert_labahn_6.5} uses linear algebra over the integers to produce a system of $n+1-r$ differential equations. The solutions for the second system in \eqref{eqn:ODE's_after_scaling} can be found from the first by integration:
\begin{equation}
    x = \exp \left(\int \frac{1}{t} F(y^{W_{\mathfrak{b}}}) \cdot V_{\mathfrak{a}} \ dt \right).
\end{equation}
Furthermore, by \cite[Theorem 4.2]{Hubert2013c} the $n+1-r$ components of $y=(t,z_1,\ldots,z_n)^{V_{\mathfrak{b}}}$ form a generating set of rational invariants of the scaling action given by $\bar{A}$. Thus, \emph{any} rational invariant of $\bar{A}$ can be written in terms of the components of $y$. Specifically, if $f(z) = \frac{p(z)}{q(z)}$ is  invariant under the scaling action defined by $\bar{A}$, then we can write this as $f(y^{W_{\mathfrak{b}}}) = \frac{p(y^{W_{\mathfrak{b}}})}{q(y^{W_{\mathfrak{b}}})}$.
\end{remark}

\subsection{Nondimensionalization} \label{sec:nondimensionalization}

Theorem~\ref{thm:hubert_labahn_6.5} gives a framework for reducing a system of differential equations by eliminating scaling symmetries. In practice, many mathematical models have more parameters than are relevant for quantitative analysis.  In these cases, there is a simpler algorithm for nondimensionalization, given by Theorem~\ref{thm:hubert_labahn_7.1}, which prioritizes reducing the number of parameters in a dynamical system.  This section presents this method following \cite[Section 7]{Hubert2013c}.

Suppose we have a dynamical model of variables $z(t) = (z_1(t),\ldots,z_q(t))$ whose dynamics involve some constants $c=(c_1,\ldots,c_p)$. The parametrized system can be written as
\begin{equation} \label{eqn:reduced_system}
    \frac{dz}{dt} = G(t,z,c) = \frac{z}{t} * F(t,z,c).
\end{equation}

Note that we can extend the system to that given by \eqref{eqn:our_ODE's} simply with the equations $\frac{dc}{dt} = 0$, so a matrix $A \in M_{r \times m}$ (where $m = 1+ q + p$) defines a scaling symmetry of \eqref{eqn:reduced_system} if and only if $F(t,z,c)$ is $\mathbb{T}_A$-invariant, which means that $F(\lambda^A * (t,z,c)) = F(t,z,c)$ for all $\lambda \in \mathbb{T}^r$. To find the scaling matrix $A$, we still use the methods of Section~\ref{sec:maximal_scaling_symmetry}. 

Assume that the Hermite multiplier $V$ of $A$ has the form%
\footnote{
    For implementation, here are the sizes of the matrices.  Since $V \in M_{r \times m}(\mathbb{Z})$ where $m = 1+q+p$, we see that $V_a \in M_{p \times r}(\mathbb{Z})$, $V_{\bar{v}} = [V_t \ \ V_v] \in M_{p \times (q+1)}(\mathbb{Z})$ and $V_c \in M_{p \times (m-r)}(\mathbb{Z})$. In $W$, we make no extra assumption here on the elements of the first $r$ rows. The bottom $p-r$ rows of $W$ are given by some matrices $W_{\bar{v}} = [W_t \ \ W_v] \in M_{(p-r) \times (q+1)}(\mathbb{Z})$ and $W_c \in M_{(p-r) \times p}(\mathbb{Z})$. Here, $V_t$ and $W_t$ are respectively the first column of $V_{\bar{v}}$ and $W_{\bar{v}}$.
}
\begin{equation} \label{eqn:assumed_hermite_normal_form}
    V = \left[\begin{array}{c|c c}
        0 & I_{q+1} & 0 \\
        V_a & V_{\bar{v}} & V_c
    \end{array}\right], \quad \quad V^{-1} = W = \left[\begin{array}{c c}
        * & * \\
        \hline
        I_{q+1} & 0 \\
        W_{\bar{v}} & W_c
    \end{array}\right].
\end{equation}

There are two main reasons for partitioning $V$ in this manner. Empirically, as mentioned in \cite[Section 7.1]{Hubert2013c}, many models in mathematical biology fit this partition.  Heuristically, this representation of $V$ ensures that we reduce only the number of constant parameters $c_1,\ldots,c_p$.

To see this, we dissect the partition of $V$ and $W$ in \eqref{eqn:assumed_hermite_normal_form}. The identity submatrix $I_{q+1}$ in $V$ gives us that our new non-constant invariants $y_0,\ldots,y_q$ are precisely $y_0=c^{u_0}t$ and $y_i=c^{u_i}z_i$, where $u_0,\ldots,u_q$ are the columns of $V_{\bar{v}}$. Furthermore, the left- and right-hand zero submatrices in $V$ respectively ensure that our $r$ auxiliary variables and $p-r$ new constants are functions only of the original constants $c=(c_1,\ldots,c_p)$. Using Theorem~\ref{thm:hubert_labahn_6.5}, we can find our invariants explicitly:
\begin{gather}
    y = [t, z, c]^{V_{\mathfrak{b}}} 
      = [c^{V_t}t, c^{V_v}*z, c^{V_c}] 
      = [\tau, \nu, \kappa], 
         \quad \textrm{so that}  \nonumber \\
    \tau \coloneqq c^{V_t}t, \quad \nu = (\nu_1,\ldots,\nu_q) \coloneqq c^{V_v} * z, \quad \kappa = (\kappa_1,\ldots,\kappa_{p-r}) \coloneqq c^{V_c}. \label{eqn:rescaled_parameter_names}
\end{gather}

\begin{theorem} \label{thm:hubert_labahn_7.1}
\cite[Section 7.1]{Hubert2013c}.
Let $W$ be as in \eqref{eqn:assumed_hermite_normal_form}. The reduced system of \eqref{eqn:reduced_system} has one time variable, $q$ dependent variables, and $p-r$ parameters respectively given by $\tau,\nu$, and $\kappa$ in \eqref{eqn:rescaled_parameter_names}. The reduced system is obtained by the substitution 
\begin{equation}
    t \mapsto \kappa^{W_t} \tau, \quad z \mapsto \kappa^{W_z} * \nu, \quad c \mapsto \kappa^{W_c}, 
\end{equation}
which gives 
\begin{equation} \label{eqn:system_after_substitution}
    \frac{d\nu}{d\tau} = \frac{\nu}{\tau} * F(\kappa^{W_t} \tau,\kappa^{W_z} * \nu,\kappa^{W_c}).
\end{equation}
\end{theorem}

\paragraph{Generalization to broader forms of Hermite multiplier} \label{sec:generalization_hermite_mul}

We generalize the form of the column Hermite multiplier $V$ given in \eqref{eqn:assumed_hermite_normal_form}, so that we may be able to apply our algorithm to cases in which $V$ does not contain the submatrix $I_{q+1}$ but instead has the more general form
\begin{equation}
    V = \left[\begin{array}{c|cc}
    0 & D & 0 \\
    V_{\mathfrak{a}} & V_{\bar{v}} & V_c
    \end{array}\right], \quad \textrm{where} \quad
    D = \textrm{diag}(d_0, d_1, \ldots, d_q).
\end{equation}
The diagonal elements $d_0,d_1,\ldots,d_q$ are any positive integers. If $d_i > 1$ for some\footnotemark \footnotetext{The case where $d_0 > 1$ works in the exact same way, by just replacing $z_i$ with $t$ and $\nu_i$ with $\tau$.} $1 \leqslant i \leqslant q$ (so that we are no longer in the setting of \eqref{eqn:assumed_hermite_normal_form}), then by Theorem~\ref{thm:hubert_labahn_6.5} we have an invariant $\nu_i=c_1^{a_1} \ldots c_p^{a_p}z_i^{d_i}$, where $a_1,\ldots,a_p$ are the elements of the $i$-th column of $V_{\bar{v}}$. We then modify the parameters as:\footnotemark
\footnotetext{Here we need to assume that our parameters $c_1,\ldots,c_p$ are real and positive, but this is often the case for mathematical models.}
\begin{equation}
    \Tilde{c}_i \coloneqq \begin{cases} \sqrt[\leftroot{-3}\uproot{3}d_i]{c_i} \ \textrm{if} \ a_i \neq 0 \\
    \ \ c_i \quad \textrm{if} \ a_i=0
    \end{cases}
    \quad \textrm{and substitute} \quad
    c_i \mapsto \begin{cases}
    \Tilde{c}_i^{d_i} \ \textrm{if} \ a_i \neq 0 \\
    \Tilde{c}_i \ \ \ \textrm{if} \ a_i=0.
    \end{cases}
\end{equation}

Under the above substitution, our new system is equivalent. The original invariant becomes $(\Tilde{c}_1^{a_1} \ldots \Tilde{c}_p^{a_p} z_i)^{d_i}$, which means $\Tilde{c}_1^{a_1} \ldots \Tilde{c}_p^{a_p}z_i$ is also invariant and so we choose $\nu_i = \Tilde{c}_1^{a_1} \ldots \Tilde{c}_p^{a_p}z_i$. Substituting for each $d_i > 0$ transforms the Hermite multiplier $V$ into $D \mapsto I_{q+1}$, and so we have returned to the setting of Theorem~\ref{thm:hubert_labahn_7.1}.

\begin{example}
Consider the differential equation
\begin{equation} \label{eqn:example:cubic_diffeq}
    \frac{dz}{dt} = \frac{cz^3}{t}
\end{equation}
where $c \in \mathbb{R}_{>0}$. The scaling matrix is $A = [0 \ \ 1 \ -3]$, whose column Hermite multiplier $V$ has submatrix $D$ with $d_0 = 1$ and $d_1 = 3$. Thus we substitute $c \mapsto \Tilde{c}^3$ and have
\begin{equation*}
    \frac{dz}{dt} = \frac{\Tilde{c}^3 z^3}{t}
\end{equation*}
which has scaling matrix $A' = [0 \ \ 1 \ -1]$. The column Hermite multiplier $V'$ of $A'$ has form \eqref{eqn:assumed_hermite_normal_form}, so Theorem~\ref{thm:hubert_labahn_7.1} gives the reduced system:
\begin{equation*}
    \frac{d\nu}{d\tau} = \frac{\nu^3}{\tau}.
\end{equation*}

\end{example}

%--------------------------------------------------%

\subsection{Practical considerations for physical systems} \label{sec:practical_considerations}

This section formalizes techniques for imposing initial conditions and functions of parameters in our nondimensionalization. We build upon this method to develop an algorithm which provides a reduced dynamical system in terms of desired invariants.

\subsubsection{Initial conditions}  \label{sec:initial_conditions}

We can incorporate an initial condition $z_i(0)=z_{i0}$ into a nondimensionalization by requiring $\frac{z_i}{z_{i0}}$ to be a scaling invariant. 

Add the constant $z_{i0}$ to the end of the variable order. To ensure that any scaling acts on $z_{i0}$ as it does on $z_i$, require $\frac{z_i}{z_{i0}}$ to be a rational invariant of our scaling action. In other words, add a column to the exponent matrix:
\begin{equation}
    K \longmapsto \left[K \ \ K_{\frac{z_i}{z_{i0}}} \right] \label{eqn:modify_exponent_matrix}
\end{equation}

Then calculate the maximal scaling matrix $A$ and its column Hermite multiplier $V$ as before.  Any rational invariant can be written in terms of the components of $y = (y_1,\ldots,y_{n+1-r}) = (t,z,c)^{V_{\mathfrak{b}}}$ (where $V_{\mathfrak{b}}$ is of the form given in \eqref{eqn:assumed_hermite_normal_form}), via the substitution $z \mapsto y^{W_{\mathfrak{b}}}$. So, we can use Theorem~\ref{thm:hubert_labahn_7.1} to obtain the corresponding constraint on $\frac{z_i}{z_{i0}}$ via substitution, as is done for the rest of the reduced system.

For example, in Section~\ref{sec:michaelis_menten_comparison_with_classical} the invariant $u=\frac{s}{s_0}$ is obtained from the reduced system via the substitution $s \mapsto u$ as given in \eqref{eqn:michaelis_menten_invariants}.

\subsubsection{Functions of constant parameters}   \label{sec:functions_of_constant_parameters}

As with both the Michaelis constant $K_m$ in Section~\ref{sec:michaelis_menten_comparison_with_classical} and $\epsilon$ in Section~\ref{sec:application_segel_slemrod}, we can incorporate rational functions of parameters into our algorithm.

Suppose we have $\ell=f(c_1,\ldots,c_p)$, a rational function of {\em constants}, which we wish to substitute into our dynamical system. First, add $\ell$ to the variable order.  Second, any scaling symmetry should act on $\ell$ in the same way it acts on $f$, so $\frac{f}{\ell}$ must be a rational invariant. We incorporate this into our framework by modifying $K$:
\begin{equation}
    K \longmapsto \left[K \ \ K_{\frac{f}{\ell}} \right].
\end{equation}

Calculate the new maximal scaling matrix $A$ and its column Hermite multiplier $V$ as before, and $W=V^{-1}$ can be used as in Theorem~\ref{thm:hubert_labahn_7.1} to give the reduced system involving the chosen invariant $\ell$.

\subsection{Completing partial sets of invariants}  \label{sec:completing_partial_sets_of_invariants}

We have a lot of flexibility in choosing our rational invariants, since we may perform any admissible column operation on $V_{\mathfrak{b}}$ to change our invariants while retaining $V$ as a valid Hermite multiplier. Of course, these operations also affect $W$. However, some parameter combinations may be preferable for analysis of a given dynamical system, so we are led to ask the following question: can we \emph{first} choose the parameter combinations we prefer, and \emph{then} see if they are invariants in which we can rewrite our dynamical system?

\subsubsection{Extending a choice of invariants}   \label{sec:extending_choice_invariants}

Suppose we are in the setting of Section~\ref{sec:nondimensionalization}, with a chosen order on our $m=q+p+1$ variables and a scaling matrix $A$ which gives rise to a Hermite multiplier $V$ with inverse $W$. We then have the following useful result: 
\begin{lemma} \label{lemma:kerA_in_kerWb}
\begin{equation*}
    \ker A \subseteq \ker W_{\mathfrak{b}}.
\end{equation*}
\end{lemma}

\begin{proof}
By Theorem~\ref{thm:normal_hermite_mult}, the columns of $V_{\mathfrak{b}}$ form a basis for the kernel of $A$, so $x \in \ker A$ can be expressed as a linear combination of the columns of $V_{\mathfrak{b}}$. But since $W=V^{-1}$ we see from the decomposition \eqref{eqn:hermite_multiplier_decomp} that $W_{\mathfrak{a}}V_{\mathfrak{b}}=0$. Thus, $x \in \ker W_{\mathfrak{a}}$.
\end{proof}

If we are given $s$ total preferred parameter combinations, then we can express these as a matrix $P \in M_{m \times s}(\mathbb{Z})$. The following theorem gives us an answer to our question at the beginning of Section~\ref{sec:completing_partial_sets_of_invariants}:

\begin{theorem} \label{thm:when_can_rewrite_using_parameter_combinations}
\emph{A given dynamical system as in \eqref{eqn:reduced_system} can be rewritten in terms of the parameter combinations determined by $P$ if and only if:}

\begin{enumerate}
    \item $AP=0$
    \item \emph{We can extend }$W_{\mathfrak{b}}P$\emph{ to a unimodular integer matrix}
\end{enumerate}

\end{theorem}

\begin{proof}
We begin with the forwards implication. If we can express the reduced dynamical system in terms of the parameter combinations given by $P$, then these combinations are necessarily invariants of the scaling action determined by $A$, which precisely means that $AP=0$.

Remark~\ref{remark:hubert_labahn6.5_rational_invariants_by_substitution} notes that the components of $y=[t \ \ z \ \ c]^{V_{\mathfrak{b}}}$ form a generating set of invariants of the $\mathbb{T}_A$-action. Thus, the new invariants $[t \ \ z \ \ c]^P$ can be written as monomials of the original invariants (the components of $y$). In other words, the invariants of $P$ are linear combinations of the columns of $V_{\mathfrak{b}}$.  Thus, there exist integer matrices $C$ and $E$ such that
\begin{equation*}
     V_{\mathfrak{b}}C=[P \ \ E]
\end{equation*}
where $E$ `extends' $P$ to a complete set of invariants. Since rewriting a dynamical system as such is reversible, $C$ must be unimodular. We apply $W_{\mathfrak{b}}$ to obtain
\begin{equation*}
    W_{\mathfrak{b}}V_{\mathfrak{b}}C = C = [W_{\mathfrak{b}}P \ \ W_{\mathfrak{b}}E].
\end{equation*}
The first equality holds because $W=V^{-1}$ and so $W_{\mathfrak{b}}V_{\mathfrak{b}}$ is the identity matrix. Since $C$ is unimodular, we have extended $W_{\mathfrak{b}}P$ to a unimodular matrix.

Conversely, let $C=[W_{\mathfrak{b}}P \ \ \tilde{E}]$ be a unimodular extension of $W_{\mathfrak{b}}$ for some matrix $\tilde{E}$. Define $E \coloneqq V_{\mathfrak{b}}\tilde{E}$, then $\tilde{E}=W_{\mathfrak{b}} V_{\mathfrak{b}}\tilde{E}=W_{\mathfrak{b}}$ since $W_{\mathfrak{b}}V_{\mathfrak{b}}=I_{m-r}$ is the identity matrix. Thus $C=W_{\mathfrak{b}}[P \ \ E]$, which means that
\begin{equation*}
    V_{\mathfrak{b}}C=V_{\mathfrak{b}}W_{\mathfrak{b}}[P \ \ E] = (I_m - V_{\mathfrak{a}}W_{\mathfrak{a}})[P \ \ E] = [P \ \ (I_m - V_{\mathfrak{a}}W_{\mathfrak{a}})E].
\end{equation*}
In the first equality we use the relationship $V_{\mathfrak{a}}W_{\mathfrak{a}}+V_{\mathfrak{b}}W_{\mathfrak{b}}=I_m$ and in the second equality we  use $AP=0$ (by assumption) and Lemma~\ref{lemma:kerA_in_kerWb}. Since post-multiplying $V_{\mathfrak{b}}$ by the unimodular matrix $C$ corresponds simply to a series of column operations of $V_{\mathfrak{b}}$, we have obtained a new Hermite multiplier $V_{\mathfrak{b}}C$. By  Theorem~\ref{thm:hubert_labahn_7.1}, this gives a reduced dynamical system in terms of the invariants determined by $P$.
\end{proof}

\subsubsection{Smith normal form}   \label{sec:smith_normal_form}

Before we present the extension algorithm, we require a tool from linear algebra over the integers.

\begin{definition}\cite[Section 2]{morandi}
An integer matrix $S \in M_{a \times b}(\mathbb{Z})$ is in \emph{Smith normal form} if $S_{ij}=0$ for $i \neq j$ and the non-zero diagonal elements $a_1,\ldots,a_k$ (in order of appearance in $S$) satisfy $a_i$ divides $a_{i+1}$ for $1 \leqslant i \leqslant k-1$.
\end{definition}

\begin{theorem}\cite[Theorem 2.1]{morandi}
\emph{For any integer matrix $X \in M_{a \times b}(\mathbb{Z})$ there exist unimodular matrices $U \in M_a(\mathbb{Z})$ and $U' \in M_b(\mathbb{Z})$ such that $UXU'$ is in Smith normal form.}
\end{theorem}

We now prove a technical result for the purpose of the algorithm in Section~\ref{sec:extension_algorithm}:

\begin{lemma}
For matrices $W_{\mathfrak{b}} \in M_{(m-r) \times m}(\mathbb{Z})$ and $P \in M_{m \times s}(\mathbb{Z})$ as defined above, if the Smith normal form decomposition
\begin{equation*}
    U(W_{\mathfrak{b}}P)U' = S
\end{equation*}
has all diagonal entries equal to $1$, then we can extend $W_{\mathfrak{b}}P$ to a unimodular matrix $C = [W_{\mathfrak{b}}P \ \ \tilde{U}]$, where $\tilde{U}$ is given by the last $m-s$ columns of $U^{-1}$.
\end{lemma}

\begin{proof}
We first note that since $s$ is the number of invariants defined by $P$, we necessarily have that $s \leqslant m-r$. By assumption, $S$ is of the form
\begin{equation*}
S = \begin{bmatrix} I_s \\ 0 \end{bmatrix}
\end{equation*}
and thus has $m-r-s$ zero-rows. Extend $S$ to $I_{m-r}$ with the matrix
\begin{equation*}
    T = \begin{bmatrix} 0 \\ I_{m-r-s} \end{bmatrix}.
\end{equation*}
Noting that $U,U'$ are invertible, we compute
\begin{align*}
    C \coloneqq& U^{-1} I_{m-r} \begin{bmatrix} U' & 0 \\ 0 & I_{m-r-s} \end{bmatrix}^{-1}  
    = \begin{bmatrix} W_{\mathfrak{b}}PU' & U^{-1}T \end{bmatrix} \begin{bmatrix} U' & 0 \\ 0 & I_{m-r-s} \end{bmatrix}^{-1} \\
    =& \begin{bmatrix} W_{\mathfrak{b}}P & U^{-1}T \end{bmatrix}.
\end{align*}
By the definition of $T$, indeed $U^{-1}T$ is given by the last $m-s$ columns of $U^{-1}$. Since both $U$ and $U'$ are unimodular we see that indeed $C=[W_{\mathfrak{b}}P \ \ \tilde{U}]$ is unimodular, where $\tilde{U} \coloneqq U^{-1}T$.

\end{proof}

\subsubsection{Extension algorithm}   \label{sec:extension_algorithm}

With these tools at our disposal, we may now summarize an algorithm for writing the reduced system of \eqref{eqn:reduced_system} in terms of our $s$ chosen invariants defined by $P \in M_{m \times s}(\mathbb{Z})$:

\begin{enumerate}
    \item Check that $AP=0$. If this does not hold, then within $P$ we have chosen non-invariants of $A$. To remedy this, choose $\tilde{P}=V_{\mathbb{b}}W_{\mathbb{b}}P$ instead.\footnote{By Lemma~\ref{lemma:kerA_in_kerWb} we have $A\tilde{P}=0$ and $W_{\mathfrak{b}}\tilde{P}=W_{\mathfrak{b}}P$, so $P$ is altered without losing any information.}
    
    \item Compute the Smith normal form decomposition $U(W_{\mathfrak{b}}P)U' = S$ of $W_{\mathfrak{b}}P$. If the leading diagonal of $S$ has an entry not equal to $1$ then $S$ is not unimodular\footnote{The non-zero entries of $S$ are unique up to multiplication by $\pm 1$ and so we can choose $U$ and $U'$ to ensure that all entries of $S$ are non-negative.}, thus we have an invalid choice of invariants and cannot continue.
    
    \item Let $\tilde{U}$ be the last $m-s$ columns of $U^{-1}$ and $C = [W_{\mathfrak{b}}P \ \ \tilde{U}]$ be an extension of $W_{\mathfrak{b}}P$ to the unimodular matrix $C$. The desired Hermite multiplier $\tilde{V}$ is
    \begin{equation*}
        \tilde{V}=V \begin{bmatrix} I_r & 0 \\ 0 & C \end{bmatrix}.
    \end{equation*}
    
    \item (Optional) $\tilde{V}_{\mathfrak{b}}$ now has the required form with the first $s$ columns equal to $P$. We can however perform column operations on the last $m-r-s$ columns of $\tilde{V}$ to ensure these columns are in Hermite normal form such that they are reduced with respect to earlier columns.
\end{enumerate}

\subsection{Sensible change of variables never lead to additional symmetries}\label{sec:Change_of_variables} 

In this section we show that if we require that in any change of variables only quantities with the same units are added together, then the maximal scaling symmetry has the same rank before and after the change of variables.

First, we consider an example which at first makes it look like a change of variable can yield more scaling symmetries, and hence lead to greater model reduction.

\begin{example}\label{eg:ChangeVar}
    Consider the following ODE system
    \[\frac{dx}{dt}=(a+b)x+b,\]
   describing the behavior of the quantity \(x\) in time, depending on the constant parameters \(a\) and \(b\). When we choose the variable order \((t,x,a,b)\), we find the scaling matrix 
   \[A=\begin{pmatrix}1 & 0 & -1& -1\end{pmatrix}.\]
   That is, we find a one-dimensional scaling symmetry.

   On the other hand, if we do the invertible change of variables $c=a+b$,
   so we have the system
   \[\frac{dx}{dt}=cx+b,\]
   and naively apply the method with variable order \((t,x,b,c)\), then we find the scaling matrix
   \[B=\begin{pmatrix}1 & 0 &-1 &-1 \\ 0 & 1 & 1&0\end{pmatrix}.\]
   This gives a two-dimensional scaling symmetry --- a seemingly greater reduction than the original system. 
   This excessive reduction results from our failure to take into account where the new parameter \(c\) came from. Indeed, noting that \(c=a+b\) tells us that, if this change of variable was dimensionally consistent, then \(a\) and \(b\) must have the same units and so \(b\) and \(c\) must also have the same units. If we add this constraint (in the form of an additional column for the matrix \(K\), corresponding to the rational function \(b/c\)), then we find scaling matrix
   \[C=\begin{pmatrix} 1 & 0 &-1 &-1\end{pmatrix}.\]
   That is, we find a one-dimensional scaling symmetry as in the original system.
\end{example}

Suppose we have an ODE system where for simplicity we denote all the quantities involved by \(x_1,\ldots,x_n\). These may be the time variable, the dependent variables, parameters, etc. Suppose we have another ODE system involving the quantities \(y_1,\ldots,y_m\). We say that \(\phi\) is a \emph{(rational) change of variable from the first system in the quantities \(x_i\) to the second system in the quantities \(y_j\)} if there are rational functions \(\phi_1, \ldots,\phi_j\) in the \(x_i\)'s such that the ODE's of the first system can be rewritten in terms of the \(\phi_j\)'s, and the identification of \(\phi_j\) with \(y_j\) give us the equations in the second system. We say such a change of variable is \emph{invertible} if the \(x_i\)'s can be written as rational functions of the \(\phi_j\)'s. We say that a scaling symmetry ensures that the change of variable is \emph{dimensionally consistent}, if it induces a scaling of the \(\phi_j\)'s.

\begin{example}
    In Example \ref{eg:ChangeVar}, the quantities involved in the original system are \(x_1=t, x_2=x, x_3=a, x_4=b\) and in the second system the quantities involved are \(y_1=t, y_2=x, y_3=b, y_4=c\). The change of variables from the first system to the second which we describe via \(c=a+b\), is given by \(\phi_1=x_1, \phi_2=x_2, \phi_3=x_4, \phi_4=x_3+x_4\). The first system is then
    \[\frac{dx_2}{dx_1}=(x_3+x_4)x_2+x_3,\]
    and substituting the expressions \(\phi_j\) with \(y_j\) we indeed get the second system
    \[\frac{dy_2}{dy_1}=(y_4)y_2+y_3.\]
    As \(x_3=(x_3+x_4)-x_4=\phi_4-\phi_3\), we can see that this change of variables is invertible.  This change of variables is dimensionally consistent with any scaling symmetry whose scaling matrix have the same third and fourth column, that is, any scaling symmetry that correspond to \(a\) and \(b\) having the same units. 
\end{example}

We can now formulate the main result of this subsection:
\begin{theorem}\label{thm:ChangeOfVariables}
Suppose that two systems of ODE's are linked by an invertible change of variables. Then the rank of the maximal scaling symmetries of the system which also ensure that the change of variables is dimensionally consistent is the same no matter which system is considered.
\end{theorem}

We obtain this result as a consequence of an algebraic lemma. Because the dynamical system does not really play a role here beyond providing the list of rational functions required to be invariant, we can rephrase the problem in terms of two fields of rational functions in the same number of indeterminates, say \(\kk(x):=\kk(x_1,\ldots,x_n)\) and \(\kk(y):=\kk(y_1,\ldots,y_n)\). The change of variables from \(\kk(x)\) to \(\kk(y)\) then yields a \(\kk\)-algebra homomorphism 
\(\phi\colon \kk(y)\to \kk(x)\), by setting \(\phi(y_j):=\phi_j\) and then we get the image of an arbitrary element by extending algebraically\footnote{We get the image of any element of \(\kk(y)\) by assuming that for any \(f,g\in \kk(y)\) and \(r\in\kk\),  \(\phi(r)=r\), \(\phi(f+g)=\phi(f)+\phi(g)\), \(\phi(fg)=\phi(f)\phi(g)\), and \(\phi(f/g)=\phi(f)/\phi(g)\).}. The inverse change of variables from \(\kk(y)\) to \(\kk(x)\) similarly yields a \(\kk\)-algebra homomorphism \(\psi\colon \kk(x)\to \kk(y)\) by setting \(\psi(y_j):=\psi_j\). 
We assume that the change of variables from \(\kk(x)\) to \(\kk(y)\) is invertible, that is,  each \(x_i\) can be written as a rational function of the \(\phi(y_j)\)'s and each \(y_j\) can be written as a rational function of the \(\psi(x_i)\)'s. This corresponds to assuming that the two \(\kk\)-algebra homomorphisms are mutually inverse, that is, \(\phi\circ \psi\) is the identity on \(\kk(x)\)  and \(\psi\circ\phi\) is the identity on \(\kk(y)\). The algebraic lemma we need to prove in order to prove Theorem \ref{thm:ChangeOfVariables} is then as follows.

\begin{lemma}\label{lem:AlgLemma}
 With the notation described above, we have:
    \begin{enumerate}
        \item The maximal scaling symmetry which allows for the change of variable from \(\kk(x)\) to \(\kk(y)\) to be dimensionally consistent has the same rank as the maximal scaling symmetry which allows for the change of variable from \(\kk(y)\) to \(\kk(x)\) to be dimensionally consistent.
    \item If \(F_1,\ldots,F_s\in \kk(x)\) and \(H_l=\psi(F_l)\), then the maximal scaling symmetry which fixes \(F_1,\ldots,F_s\) and allows for the change of variable from \(\kk(x)\) to \(\kk(y)\) to be dimensionally consistent has the same rank as the maximal scaling symmetry which fixes \(G_1,\ldots,G_s\) and  allows for the change of variable from \(\kk(y)\) to \(\kk(x)\) to be dimensionally consistent.
     \end{enumerate}
\end{lemma}

\begin{proof}[Proof of Theorem \ref{thm:ChangeOfVariables}]
    The theorem follows immediately from the second statement of Lemma~\ref{lem:AlgLemma}.
\end{proof}

\begin{proof}[Proof of Lemma \ref{lem:AlgLemma}]
We write \(\TT_x\) to denote the \(n\)-dimensional torus acting on \(\kk(x)\) by scaling the \(n\) indeterminates independently, and \(\TT_y\) for the \(n\)-dimensional torus acting on \(\kk(y)\). The maximal scaling symmetry of \(\kk(x)\) which allows the change of variables from \(\kk(x)\) to \(\kk(y)\) to be dimensionally consistent is the maximal subtorus of \(\TT_x\) such that \(\phi(y_j)\) is homogeneous for all \(j\), we denote it by \(\TT_\phi\). In other words, \(\TT_\phi\) is the set of all \(\lambda\in \TT_x\) such that the scaling of the \(x_i\)'s induces a scaling of the \(\phi(y_j)\)'s, i.e., for each \(j\) there is a \(\beta_j\in\ZZ^n\) such that \(\lambda \star\phi(y_j)=\lambda^{\beta_j}\phi(y_j)\).
Similarly, \(\TT_\psi\)  denotes the  maximal scaling symmetry of \(\kk(y)\) which allows the change of variables from \(\kk(y)\) to \(\kk(x)\) to be dimensionally consistent.

For each \(i\) and \(j\), \(\phi(y_j)=\phi_j\) and \(\psi(x_i)=\psi_i\)  are rational functions in the \(x_i\)'s and \(y_j\)'s, respectively, with coefficients in \(\kk\), which means that we can write
\begin{align*}
   \phi(y_j)=\frac{\sum{\phi_{\delta_j}x^{\delta_j}}}{\sum{\phi_{\tilde{\delta}_j}x^{\tilde{\delta}_j}}}, \text{ and }&
   \psi(x_i)=\frac{\sum{\psi_{\gamma_i}y^{\gamma_i}}}{\sum{\psi_{\tilde{\gamma}_i}y^{\tilde{\gamma}_i}}},
\end{align*}
where only finitely many \(\phi_{\delta_j},\phi_{\tilde{\delta}_j},\psi_{\gamma_i},\psi_{\tilde{\gamma}_i}\in \kk\) are non-zero and \({\delta_j},{\tilde{\delta}_j},{\gamma_i},{\tilde{\gamma}_i}\in \ZZ^n\).

We construct a homomorphism of groups \(\rho\colon \TT_\phi \to \TT_y\). Take \(\lambda\in\TT_\phi\) and pick \(1\leq j\leq n\). As mentioned earlier, there is \(\beta_j\in\ZZ^n\) such that \(\lambda \star \phi(y_j)=\lambda^{\beta_j}\phi(y_j)\). Setting \(\lambda \bullet y_j:=\lambda^{\beta_j}y_j\) for each \(j\) gives a well defined scaling of the indeterminates \(y_1,\ldots,y_j\), and so mapping \(\lambda\) to the corresponding element of \(\TT_y\) gives a well-defined group homomorphism. Note that as \(\psi\) and \(\phi\) are mutually inverse, we have
        \begin{align*}
            x_i &= \phi(\psi(x_i))
                   =\phi\left(\frac{\sum{\psi_{\gamma_i}y^{\gamma_i}}}{\sum{\psi_{\tilde{\gamma}_i}y^{\tilde{\gamma}_i}}}\right)
                   =\frac{\sum{\psi_{\gamma_i}\phi(y_1)^{\gamma_{i,1}}}\cdots\phi(y_n)^{\gamma_{i,n}}}{\sum{\psi_{\tilde{\gamma}_i}\phi(y_1)^{\tilde{\gamma}_{i,1}}}\cdots\phi(y_n)^{\tilde{\gamma}_{i,n}}}.       
        \end{align*}
        Take \(\lambda\in\TT_\phi\) and suppose that \(\rho(\lambda)\) is the identity element in \(\TT_y\), that is, suppose that \(\lambda\bullet y_j=y_j\) for all \(j\). Then by definition of the action of \(\lambda\) on \(y_j\), it follows that
        \(\lambda\star \phi(y_j)=\phi(y_j)\) for all \(j\). Thus, we have
        \begin{align*}
            \lambda \star x_i &=\lambda \star \frac{\sum{\psi_{\gamma_i}\phi(y_1)^{\gamma_{i,1}}}\cdots\phi(y_n)^{\gamma_{i,n}}}{\sum{\psi_{\tilde{\gamma}_i}\phi(y_1)^{\tilde{\gamma}_{i,1}}}\cdots\phi(y_n)^{\tilde{\gamma}_{i,n}}}
                            =\frac{\sum{\psi_{\gamma_i}(\lambda\star\phi(y_1))^{\gamma_{i,1}}}\cdots(\lambda\star\phi(y_n))^{\gamma_{i,n}}}{\sum{\psi_{\tilde{\gamma}_i}(\lambda\star\phi(y_1))^{\tilde{\gamma}_{i,1}}}\cdots(\lambda\star\phi(y_n))^{\tilde{\gamma}_{i,n}}}\\
                            &=\frac{\sum{\psi_{\gamma_i}\phi(y_1)^{\gamma_{i,1}}}\cdots\phi(y_n)^{\gamma_{i,n}}}{\sum{\psi_{\tilde{\gamma}_i}\phi(y_1)^{\tilde{\gamma}_{i,1}}}\cdots\phi(y_n)^{\tilde{\gamma}_{i,n}}}
                            = x_i,
        \end{align*}
         which then implies that \(\lambda\) is the identity in \(\TT_x\) (and so also in \(\TT_\phi\)).  Therefore, the group homomorphism \(\rho\) has a trivial kernel, which means that \(\TT_\phi\) is isomorphic to a subgroup of \(\TT_y\).

         To prove the first statement, we show that \(\rho(\lambda)\in \TT_\psi\) for any \(\lambda\in\TT_\phi\). Indeed, this then implies that \(\TT_\phi\) is isomorphic to a subgroup of \(\TT_\psi\), and in particular \(\TT_\phi\) has rank at most equal to the rank of \(\TT_\psi\). The same argument can be used to show that \(\TT_\psi\) has rank at most equal to the rank of \(\TT_\phi\), proving the statement.

         Take \(\lambda \in \TT_\phi\). As \(\TT_\phi\) is a subgroup of \(\TT_x\), for each \(i\) there is an \(\alpha_i\in\ZZ^n\) such that \(\lambda\star x_i=\lambda^{\alpha_i}x_i\). Then for any \(i\), we have
         \[ 
             \lambda \bullet \psi(x_i) = \lambda \bullet  \frac{\sum{\psi_{\gamma_i}y^{\gamma_i}}}{\sum{\psi_{\gamma_i}y^{\tilde{\gamma}_i}}}\\
                                    =\frac{\sum{\psi_{\gamma_i}(\lambda\bullet y_1)^{\gamma_{i,1}}\cdots (\lambda\bullet y_n)^{\gamma_{i,n}}}}{\sum{\psi_{\tilde{\gamma}_i}(\lambda \bullet y_1)^{\tilde{\gamma}_{i,1}}\cdots (\lambda \bullet y_n)^{\tilde{\gamma}_{i,n}}}}.\]
        Using the definition of \(\lambda \bullet y_j\), this is equal to
\[\frac{\sum{\psi_{\gamma_i}(\lambda^{\beta_1} y_1)^{\gamma_{i,1}}\cdots (\lambda^{\beta_1} y_n)^{\gamma_{i,n}}}}{\sum{\psi_{\tilde{\gamma}_i}(\lambda^{\beta_1} y_1)^{\tilde{\gamma}_{i,1}}\cdots (\lambda^{\beta_1} y_n)^{\tilde{\gamma}_{i,n}}}}\\
                                     =\frac{\sum{\psi_{\gamma_i}\lambda^{\beta_1\gamma_{i,1}}( y_1)^{\gamma_{i,1}}\cdots \lambda^{\beta_1\gamma_{i,n}} (y_n)^{\gamma_{i,n}}}}{\sum{\psi_{\tilde{\gamma}_i}\lambda^{\beta_1\tilde{\gamma}_{i,1}} (y_1)^{\tilde{\gamma}_{i,1}}\cdots \lambda^{\beta_1\tilde{\gamma}_{i,n}} (y_n)^{\tilde{\gamma}_{i,n}}}}.\]
Using that \(\psi\) and \(\phi\) are mutual inverse, this is then

                                     \[\frac{\sum{\psi_{\gamma_i}(\lambda^{\beta_1})^{\gamma_{i,1}}( \psi(\phi(y_1)))^{\gamma_{i,1}}\cdots (\lambda^{\beta_1})^{\gamma_{i,n}} (\psi(\phi(y_n)))^{\gamma_{i,n}}}}{\sum{\psi_{\tilde{\gamma}_i}(\lambda^{\beta_1})^{\tilde{\gamma}_{i,1}} (\psi(\phi(y_1)))^{\tilde{\gamma}_{i,1}}\cdots (\lambda^{\beta_1})^{\tilde{\gamma}_{i,n}} (\psi(\phi(y_n)))^{\tilde{\gamma}_{i,n}}}},\]
      and since \(\psi\) is a homomorphism this equals
      \begin{align*}\psi\left(\frac{\sum{\psi_{\gamma_i}(\lambda^{\beta_1})^{\gamma_{i,1}}( \phi(y_1))^{\gamma_{i,1}}\cdots (\lambda^{\beta_1})^{\gamma_{i,n}} (\phi(y_n))^{\gamma_{i,n}}}}{\sum{\psi_{\tilde{\gamma}_i}(\lambda^{\beta_1})^{\tilde{\gamma}_{i,1}} (\phi(y_1))^{\tilde{\gamma}_{i,1}}\cdots (\lambda^{\beta_1})^{\tilde{\gamma}_{i,n}} (\phi(y_n))^{\tilde{\gamma}_{i,n}}}}\right)\\
                                     =\psi\left(\frac{\sum{\psi_{\gamma_i}(\lambda^{\beta_1} \phi(y_1))^{\gamma_{i,1}}\cdots (\lambda^{\beta_1}\phi(y_n))^{\gamma_{i,n}}}}{\sum{\psi_{\tilde{\gamma}_i}(\lambda^{\beta_1}\phi(y_1))^{\tilde{\gamma}_{i,1}}\cdots (\lambda^{\beta_1}\phi(y_n))^{\tilde{\gamma}_{i,n}}}}\right).\end{align*}
        Next, using the definition of \(\lambda \star\phi(y_j)\) and the expression for \(x_i\), we get that the above equals
                            \begin{align*}&\psi\left(\frac{\sum{\psi_{\gamma_i}(\lambda\star \phi(y_1))^{\gamma_{i,1}}\cdots (\lambda\star\phi(y_n))^{\gamma_{i,n}}}}{\sum{\psi_{\tilde{\gamma}_i}(\lambda\star\phi(y_1))^{\tilde{\gamma}_{i,1}}\cdots (\lambda\star\phi(y_n))^{\tilde{\gamma}_{i,n}}}}\right)\\
                                     &~~~~~~~~~~~~~~~~~~~=\psi\left(\lambda\star\frac{\sum{\psi_{\gamma_i}( \phi(y_1))^{\gamma_{i,1}}\cdots (\phi(y_n))^{\gamma_{i,n}}}}{\sum{\psi_{\tilde{\gamma}_i}(\phi(y_1))^{\tilde{\gamma}_{i,1}}\cdots (\phi(y_n))^{\tilde{\gamma}_{i,n}}}}\right)
                                     =\psi(\lambda\star x_i)=\psi(\lambda_i^{\alpha_i}x_i).\end{align*}
        Finally, as \(\psi\) is a homomorphism, we have shown that
        \[\lambda \bullet \psi(x_i)=\psi(\lambda_i^{\alpha_i}x_i)=\lambda_i^{\alpha_i}\psi(x_i),\]
         which means that \(\rho(\lambda)\in\TT_\psi\) as desired. 

         By symmetry again, to show the second statement it suffices to show that if \(\lambda \in\TT_\phi\) fixes \(F_l\) for all \(l\), then \(\rho(\lambda)\) fixes \(H_l\) for all \(l\). Fix \(l\). Take  \(\lambda \in\TT_\phi\) and assume that it fixes \(F_l\). Note that the long calculation at the end of the proof of the first statement implies that \(\lambda\bullet \psi(F')=\psi(\lambda\star F')\) for any \(F'\in\kk(x)\). Hence, we have that
         \[
             \lambda\bullet H = \lambda\bullet(\psi(F))= \psi(\lambda\star F)=\psi(F)=H.\]
    
\end{proof}

\section{Applications} \label{sec:applications}

First, we illustrate how the n\"aive application of Theorem~\ref{thm:hubert_labahn_7.1} to Michaelis-Menten kinetics without considering initial conditions or familiar constants leads to a non-standard nondimensionalization.  We then study in full detail how we can recover the familiar nondimensionalizations from \cite[Chapter 6]{murray}, subsequently considering the quasi-steady-state approximation of Segel and Slemrod \cite[Section 6.2]{murray}.  We also present new nondimensionalizations of both two- and four-variable Michaelis-Menten.  The last two examples in this section are a cell cycle control network, and a linear compartment model for vaccine injection.

\subsection{Michaelis-Menten kinetics}
\label{sec:example_michaelis_menten}

Two chemical species $E$ and $S$ reversibly combine to form $ES$, which decomposes irreversibly into $E$ and $P$ \cite{murray}.
\[
E+S\leftrightarrow ES \rightarrow E+P
\]
Incorporating conservation relations, we translate the reactions into an ODE model:
\begin{align}  \label{eqn:michaelis_menten_kinetics}
\frac{ds}{dt} = -k_1(e_0 - c)s + k_{-1}c  \qquad
\frac{dc}{dt} = k_1(e_0 - c)s - k_{-1} c - k_2 c. 
\end{align}
Variables $s$ and $c$ are respectively the concentrations of substrate $S$ and complex $ES$.  The parameters $k_1, k_2, k_{-1}$ represent rate constants, and $e_0$ represents a conserved quantity.  This model has one independent and two dependent variables, and four parameters.  The goal is to maximally reduce the number of parameters in the model.

\subsubsection{Compute the maximal scaling action} 

\paragraph{Specify an order} 
We chose the order $(t, s, c, k_{-1}, k_2, k_1, e_0)$.
For each variable $z_i$, rewrite the system in the form $\frac{dz_i}{dt} = \frac{z_i}{t}F_{z_i}$:
\begin{align*} 
\frac{ds}{dt} &= \frac{s}{t} \left(\frac{-k_1e_0ts + k_1tcs + k_{-1} t c}{s}\right)  \\
\frac{dc}{dt} &= \frac{c}{t} \left(\frac{k_1(e_0 - c)ts - k_{-1}tc - k_2 tc}{c}\right) %\label{eq:motivational_example_K_2}
\end{align*}

\paragraph{Calculate the exponent matrix $M = [M_{F_{z_i}}]$} 
Form exponent matrices $M_{F_i}$ by writing the exponents for each monomial in the fully expanded expression of $F_i$ as column vectors.  Pick a column and subtract it from the others.
%Note that constant terms correspond to $0$ vectors.
%
\begin{equation*}
\begin{blockarray}{cccccc}
 & \scalemath{0.7}{-k_1e_0ts} & \scalemath{0.7}{k_1tcs} & \scalemath{0.7}{-k_{-1} t c}& \scalemath{0.7}{s} \\
\begin{block}{c[ccccc]}
\scalemath{0.7}{t}      & 1 & 1 &  1 & 0 \\
\scalemath{0.7}{s}      & 1 & 1 &  0 & 1 \\
\scalemath{0.7}{c}      & 0 & 1 &  1 & 0 \\
\scalemath{0.7}{k_{-1}} & 0 & 0 &  1 & 0 \\
\scalemath{0.7}{k_2}    & 0 & 0 &  0 & 0 \\
\scalemath{0.7}{k_1}    & 1 & 1 &  0 & 0 \\
\scalemath{0.7}{e_0}    & 1 & 0 &  0 & 0\\
\end{block}
\end{blockarray}
\mapsto
\begin{blockarray}{cccc}
 \\
\begin{block}{[cccc]}
1 & 1 & 1 \\
0 & 0 & -1 \\
0 & 1 & 1  \\
0 & 0 & 1 \\
0 & 0 & 0 \\
1 & 1 & 0 \\
1 & 0 & 0 \\
\end{block}
\end{blockarray}
= M_{F_{z_1}}
\end{equation*}
Finally, concatenate to form $M =
\left[ \begin{array}{c|c|c|c}
M_{F_{z_1}} & M_{F_{z_2}} & \ldots & M_{F_{z_n}}
\end{array} 
\right]$.

\paragraph{Calculate scaling matrix $A$} 
Perform row Hermite normal form decomposition $U\cdot M = H_M$.  Let $r$ be the number of zero-rows of $H_M$.  The bottom $r$ rows of $U$ are the scaling matrix $A$. 
\begin{equation*} \scalemath{0.84}{
\begin{array}{c@{}c@{}c@{}c@{}}
U & M & = & H_M \\
\left[ \begin{array}{ccccccc}
0 & 0 & 0 & 0 & 0 & 0 & 1 \\
0 & 0 & 1 & -1 & 0 & 0 & 0 \\
0 & 0 & 0 & 1 & 0 & 0 & 0 \\
0 & 0 & -1 & 1 & 0 & 1 & -1 \\
0 & 0 & 0 & 0 & 1 & 0 & 0 \\
\hline
1 & 0 & 0 & -1 & -1 & -1 & 0 \\
0 & 1 & 1 & 0 & 0 & -1 & 1 \\
\end{array} \right]
& 
\left[
\begin{array}{ccc|cccc}
1 & 1 & 1 & 1 & 1 & 1 & 1 \\
0 & 0 & -1 & 1 & 0 & 0 & 1 \\
0 & 1 & 1 & 0 & 0 & 0 & -1 \\
0 & 0 & 1 & 0 & 0 & 1 & 0 \\
0 & 0 & 0 & 0 & 1 & 0 & 0 \\
1 & 1 & 0 & 1 & 0 & 0 & 1 \\
1 & 0 & 0 & 0 & 0 & 0 & 1 \\
\end{array}
\right]
 &= 
&\left[ \begin{array}{ccccccc}
1 & 0 & 0 & 0 & 0 & 0 & 1 \\
0 & 1 & 0 & 0 & 0 & -1 & -1 \\
0 & 0 & 1 & 0 & 0 & 1 & 0 \\
0 & 0 & 0 & 1 & 0 & 1 & 1 \\
0 & 0 & 0 & 0 & 1 & 0 & 0 \\
\hline
0 & 0 & 0 & 0 & 0 & 0 & 0 \\
0 & 0 & 0 & 0 & 0 & 0 & 0 \\
\end{array} \right]
\end{array}
}
\end{equation*}
There is an $r=2$-dimensional scaling action, and we can reduce two variables from the system.  
The scaling matrix from the two bottom rows of $U$ is
\begin{equation*}
A = 
\begin{blockarray}{ccccccc}
\scalemath{0.7}{t} & \scalemath{0.7}{s} & \scalemath{0.7}{c} & \scalemath{0.7}{k_{-1}} & \scalemath{0.7}{k_2} & \scalemath{0.7}{k_1} & \scalemath{0.7}{e_0}  \\
\begin{block}{[ccccccc]}
1 & 0 & 0 & -1 & -1 & -1 & 0\\
0 & 1 & 1 & 0 & 0 & -1 & 1  \\  
\end{block}
\end{blockarray}.
\end{equation*}
The scaling actions are
\begin{equation*}
\begin{matrix}
    \lambda \cdot (t, s, c, k_{-1}, k_2, k_1, e_0) &\mapsto& \left( \lambda t, s, c, \lambda^{-1}k_{-1}, \lambda^{-1}k_2, \lambda^{-1}k_1, e_0 \right) \\
\mu \cdot (t, s, c, k_{-1}, k_2, k_1, e_0) &\mapsto& \left( t, \mu s, \mu c, k_{-1}, k_2, \mu^{-1} k_1, \mu e_0 \right)
\end{matrix}.
\end{equation*}
The first scaling action corresponds to the fundamental unit of time, and the second corresponds to concentration (eg, mol/l).

\subsubsection{Compute invariants}

Invariants are computed from the Hermite multiplier matrices.  Since there are infinitely many, and each gives a different nondimensionalization, there is an infinite number of permissible nondimensionalized models.  To control for this, the scaling matrix $A$ is put into {\em canonical} column Hermite normal form, $AV = H$:
\begin{align*}
\begin{array}{c@{}c@{}c@{}c@{}}
A \, \cdot\, 
&
\begin{blockarray}{cccccccc}
\begin{block}{[cc|ccccc]c}
0 & 0 & 1 & 0 & 0 & 0 & 0 & \scalemath{0.7}{t} \\
0 & 0 & 0 & 1 & 0 & 0 & 0 & \scalemath{0.7}{s} \\
0 & 0 & 0 & 0 & 1 & 0 & 0 & \scalemath{0.7}{c} \\
0 & 0 & 0 & 0 & 0 & 1 & 0 & \scalemath{0.7}{k_{-1}} \\
0 & 0 & 0 & 0 & 0 & 0 & 1 & \scalemath{0.7}{k_2} \\
-1 & 0 & 1 & 0 & 0 & -1 & -1 & \scalemath{0.7}{k_1}  \\
-1 & 1 & 1 & -1 & -1 & -1 & -1 & \scalemath{0.7}{e_0}  \\
\end{block}
 &  & \scalemath{0.7}{e_0 k_1 t} & \scalemath{0.7}{\frac{s}{e_0}} & \scalemath{0.7}{\frac{c}{e_0}} & \scalemath{0.7}{\frac{k_{-1}}{e_0 k_1}} & \scalemath{0.7}{\frac{k_2}{e_0 k_1}} & \\
\end{blockarray}
& = &
\begin{bmatrix}
1 & 0 & 0 & 0 & 0 & 0 & 0 \\
0 & 1 & 0 & 0 & 0 & 0 & 0 \\
\end{bmatrix}
\end{array}
\end{align*}
Name the invariants
\begin{equation*}
\tau = e_{0} k_{1} t,\quad u = \frac{s}{e_{0}},\quad v = \frac{c}{e_{0}},\quad c_0 = \frac{k_{-1}}{e_0 k_1},\quad c_1=\frac{k_{-1}}{e_{0} k_{1}}.
\end{equation*}

\subsubsection{Use substitutions to write the nondimensionalized system}

To re-write the system in terms of dimensionless quantities, substitute the invariants read from columns the inverse of the Hermite multiplier $V^{-1}$.
\begin{equation*}
W = V^{-1} = 
\begin{blockarray}{cccccccc}
\scalemath{0.7}{t} & \scalemath{0.7}{s} & \scalemath{0.7}{c} & \scalemath{0.7}{k_{-1}} & \scalemath{0.7}{k_2} & \scalemath{0.7}{k_1} & \scalemath{0.7}{e_0} \\
\begin{block}{[ccccccc]c}
1 & 0 & 0 & -1 & -1 & -1 & 0 & \\
0 & 1 & 1 & 0 & 0 & -1 & 1 & \\
\cline{1-7}% don't use \hline or  it won't compile, insanely enough
1 & 0 & 0 & 0 & 0 & 0 & 0 & \scalemath{0.7}{\tau} \\
0 & 1 & 0 & 0 & 0 & 0 & 0 & \scalemath{0.7}{u} \\
0 & 0 & 1 & 0 & 0 & 0 & 0 & \scalemath{0.7}{v} \\
0 & 0 & 0 & 1 & 0 & 0 & 0 & \scalemath{0.7}{c_0} \\
0 & 0 & 0 & 0 & 1 & 0 & 0 & \scalemath{0.7}{c_1} \\
\end{block}
\scalemath{0.7}{t \mapsto \tau} & \scalemath{0.7}{s \mapsto u} & \scalemath{0.7}{c \mapsto v} & \scalemath{0.7}{k_{-1} \mapsto c_0} & \scalemath{0.7}{k_2 \mapsto c_1} & \scalemath{0.7}{k_1 \mapsto 1} & \scalemath{0.7}{e_0 \mapsto 1} & 
\end{blockarray}
\end{equation*}
After making the substitutions,
the reduced nondimensionalized system is
\begin{align*}
\frac{du}{d\tau} &= c_{0} v + u v - u  & 
\frac{dv}{d\tau} &= - c_{0} v - c_{1} v - u v + u.
\end{align*}
We now have one independent variable, two dependent variables, and two parameters --- a reduction by two.

%-------------------------------------------------------------------------

\subsection{Michaelis-Menten with initial conditions and a constant} \label{sec:michaelis_menten_comparison_with_classical}

Again consider Michaelis-Menten kinetics in Section~\ref{sec:example_michaelis_menten} given by the ODE system 
\begin{equation} 
    \frac{ds}{dt} = -k_1e_0s + k_1cs + k_{-1}c \qquad \qquad
    \frac{dc}{dt} = k_1e_0s - k_1cs - k_{-1}c - k_2c,
\end{equation}
together with initial conditions $s(0)=s_0$ and \emph{Michaelis constant} $K_m = \frac{k_{-1} + k_2}{k_1}$.  To recover the usual dimensionless system as in \cite[(6.13)]{murray}, we choose variable order $(t,s,c,K_m,k_{-1},k_2,k_1,e_0,s_0)$.  The desired reduced system has the initial condition $u(0)=\frac{s(0)}{s_0}=1$, so to ensure our algorithm arrives at this same initial condition we would like $\frac{s}{s_0}$ to be a scaling invariant. Additionally, we need to impose that $\frac{k_{-1}+k_2}{k_1 K_m}$ is invariant under the maximal scaling action. Thus, we append several columns to our exponent matrix $M$,%
\footnote{In this example we use $M$ to denote the exponent matrix instead of $K$, since for the nondimensionalized Michaelis-Menten model in \cite[(6.12)]{murray} $K$ denotes the dimensionless parameter $\frac{K_m}{s_0}$.} 
\begin{equation}
    \begin{bmatrix} 0 & 1 & 0 & 0 & 0 & 0 & 0 & 0 & -1 \end{bmatrix}^\intercal, \quad \begin{bmatrix} 0 & 0 & 0 & -1 & 1 & 0 & -1 & 0 & 0 \\ 0 & 0 & 0 & -1 & 0 & 1 & -1 & 0 & 0 \end{bmatrix}^\intercal.
\end{equation}
Thus $M$ has $9$ rows (one for each variable) and $10$ columns (one for each monomial exponent vector). Compute the row Hermite normal form $UM=H$ and read off the two bottom rows of $U$ to find the maximal scaling matrix:
\begin{equation} \label{eqn:michaelis_maximal_scaling_matrix}
    A = \left[\begin{array}{ccccccccc}
        1 & 0 & 0 & 0 & -1 & -1 & -1 & 0 & 0 \\
        0 & 1 & 1 & 1 & 0 & 0 & -1 & 1 & 1 \\
    \end{array}\right]
    .
\end{equation}
Compute the column Hermite multiplier $V$ and scaling invariants:
\begin{equation}
V =
\begin{blockarray}{cccccccccc}
\begin{block}{[cc|ccccccc]c}
0 & 0 & 1 & 0 & 0 & 0 & 0 & 0 & 0 & \scalemath{0.7}{t} \\
0 & 0 & 0 & 1 & 0 & 0 & 0 & 0 & 0 & \scalemath{0.7}{s} \\
0 & 0 & 0 & 0 & 1 & 0 & 0 & 0 & 0 & \scalemath{0.7}{c} \\
0 & 0 & 0 & 0 & 0 & 1 & 0 & 0 & 0 & \scalemath{0.7}{K_m} \\
0 & 0 & 0 & 0 & 0 & 0 & 1 & 0 & 0 & \scalemath{0.7}{k_{-1}} \\
0 & 0 & 0 & 0 & 0 & 0 & 0 & 1 & 0 & \scalemath{0.7}{k_2} \\
-1 & 0 & 1 & 0 & 0 & 0 & -1 & -1 & 0 & \scalemath{0.7}{k_{1} } \\
0 & 0 & 0 & 0 & 0 & 0 & 0 & 0 & 1 & \scalemath{0.7}{e_0 } \\
-1 & 1 & 1 & -1 & -1 & -1 & -1 & -1 & -1 & \scalemath{0.7}{s_0 } \\
\end{block}
\scalemath{0.7}{\frac{1}{s_0 k_1}} & \scalemath{0.7}{s_0} & \scalemath{0.7}{s_0 k_{1} t} & \scalemath{0.7}{\frac{s}{s_{0}}} & \scalemath{0.7}{\frac{c}{s_{0}}} & \scalemath{0.7}{\frac{K_m}{s_{0}}} & \scalemath{0.7}{\frac{k_{-1}}{k_{1} s_{0}}} & \scalemath{0.7}{\frac{k_{2}}{k_{1} s_{0}}} & \scalemath{0.7}{\frac{e_{0}}{s_{0}}} \\
\end{blockarray}.
\end{equation}

This is very close to the classical Michaelis-Menten nondimensionalization.  Constants $\lambda = \frac{k_2}{k_1 s_0}$ and $K = \frac{k_{-1}+k_2}{k_1 s_0}$ from \cite[(6.12)]{murray} are recovered by $V_8$ and $V_8 + V_9$ respectively (where $V_j$ denotes the $j$-th column of $V$).

We want $\tau = k_1 e_0 t$, but have $\tau = s_0 k_1 t$ in column $3$.  We also want $v = c/e_0$ but have $c/s_0$.  So we need to modify the invariants given by columns $3$ and $5$. As noted in Section~\ref{sec:hermite_multiplier} we may perform column operations\footnotemark  on $V$ to obtain another column Hermite multiplier $V'$, thus modifying the invariants. In this problem, let $V'$ be identical to $V$ except $V'_3 = V_3 + V_9$ and $V'_5 = V_5 - V_9$. This gives
\footnotetext{To understand which column operations on $V$ are permissible here, observe the following two things. First, since $V$ is a column Hermite multiplier for $A$, we know that $AV=L$ is in column Hermite normal form. Second, since performing column operations on $V$ is equivalent to right-multiplying $V$ by a unimodular matrix $J$ (see Section~\ref{sec:hermite_multiplier}), we just have to make sure that $AVJ=LJ$ is still in column Hermite normal form. Now, by computing $AV=L$ using the $A$ and $V$ given above we see that $L_j = 0$ for $3 \leqslant j \leqslant 9$. Thus by restricting our operations to only the last seven columns of $V$ to obtain $V'$ ensures that $AV'=L$ is still in column Hermite normal form, since we have only permuted the zero columns of $L$.}
\begin{equation} \label{eqn:micheaelis_menten_hermite_multiplier}
V' =
\begin{blockarray}{cccccccccc}
\begin{block}{[cc|ccccccc]c}
0 & 0 & 1 & 0 & 0 & 0 & 0 & 0 & 0 & \scalemath{0.7}{t} \\
0 & 0 & 0 & 1 & 0 & 0 & 0 & 0 & 0 & \scalemath{0.7}{s} \\
0 & 0 & 0 & 0 & 1 & 0 & 0 & 0 & 0 & \scalemath{0.7}{c} \\
0 & 0 & 0 & 0 & 0 & 1 & 0 & 0 & 0 & \scalemath{0.7}{K_m} \\
0 & 0 & 0 & 0 & 0 & 0 & 1 & 0 & 0 & \scalemath{0.7}{k_{-1}} \\
0 & 0 & 0 & 0 & 0 & 0 & 0 & 1 & 0 & \scalemath{0.7}{k_2} \\
-1 & 0 & 1 & 0 & 0 & 0 & -1 & -1 & 0 & \scalemath{0.7}{k_{1} } \\
0 & 0 & 1 & 0 & -1 & 0 & 0 & 0 & 1 & \scalemath{0.7}{e_0 } \\
-1 & 1 & 0 & -1 & 0 & -1 & -1 & -1 & -1 & \scalemath{0.7}{s_0 } \\
\end{block}
\scalemath{0.7}{\frac{1}{s_0 k_1}} & \scalemath{0.7}{s_0} & \scalemath{0.7}{e_0 k_{1} t} & \scalemath{0.7}{\frac{s}{s_{0}}} & \scalemath{0.7}{\frac{c}{e_{0}}} & \scalemath{0.7}{\frac{K_m}{s_{0}}} & \scalemath{0.7}{\frac{k_{-1}}{k_{1} s_{0}}} & \scalemath{0.7}{\frac{k_{2}}{k_{1} s_{0}}} & \scalemath{0.7}{\frac{e_{0}}{s_{0}}} \\
\end{blockarray}
\end{equation}
We can recover the familiar nondimensionalization with a change of variables.  Apply Theorem~\ref{thm:hubert_labahn_7.1} and using \eqref{eqn:rescaled_parameter_names} name the invariants:
\begin{equation} \label{eqn:michaelis_menten_invariants}
    \tau = e_0k_1t, \quad (u,v)=\left(\frac{s}{s_0},\frac{c}{e_0}\right), \quad (K,K-\lambda,\lambda,\epsilon)=\left(\frac{K_m}{s_0},\frac{k_{-1}}{k_1 s_0},\frac{k_2}{k_1 s_0},\frac{e_0}{s_0}\right) .
\end{equation}
This recovers the nondimensionalization given in \cite[(6.12)]{murray}. Substituting
\begin{equation}
    (t,s,c,K_m,k_{-1},k_2,k_1,e_0,s_0) \longmapsto \left(\frac{\tau}{\epsilon},u,\epsilon v,K,K-\lambda,\lambda,1,\epsilon,1\right) \label{eqn:mm_substitutions}
\end{equation}
into \eqref{eqn:michaelis_menten_kinetics} to give the familiar dimensionless system from \cite[(6.13)]{murray}:
\begin{align*}
\frac{du}{d\tau} & = -u + (u+K-\lambda)v & u(0)&=1 \\
\epsilon \frac{dv}{d\tau} &= u - (u+K)v    &  v(0)&=0.
\end{align*}

\subsection{Multiple timescales: Segel-Slemrod analysis} \label{sec:application_segel_slemrod}

As explained in \cite[Section 6.2]{murray}  it need not always be the case that $\epsilon=\frac{e_0}{s_0} \ll 1$.  The Segel-Slemrod analysis remedies this by setting $\epsilon = \frac{e_0}{s_0 + K_m}$ and considering two timescales
\begin{equation}
    \tau = \frac{t}{t_c} = k_1(s_0+K_m)t, \qquad T = \frac{(1+\rho)t}{t_s} = \epsilon (1+\rho)k_2t
\end{equation}
\noindent where $\rho=\frac{k_{-1}}{k_2}$. The first timescale is applicable for $0 \leqslant \tau \ll 1$, and the second is valid outside a small neighborhood of the origin. We choose the variable order $(t,s,c,\epsilon,k_{-1},k_2,k_1,K_m,e_0,s_0)$ with our new $\epsilon=\frac{e_0}{s_0 + K_m}$. Together with \eqref{eqn:michaelis_menten_kinetics}, we impose the condition that we want $\frac{s}{s_0},\frac{k_{-1} + k_2}{k_1 K_m}$ and $\frac{e_0}{\epsilon(s_0 + K_m)}$ to be scaling invariants. The algorithm computes the column Hermite multiplier:

\begin{equation} \label{eqn:mm_column_hermite_multiplier}
V =
\begin{blockarray}{ccccccccccc}
\begin{block}{[cc|cccccccc]c}
0 & 0 & 1 & 0 & 0 & 0 & 0 & 0 & 0 & 0 & \scalemath{0.7}{t} \\
0 & 0 & 0 & 1 & 0 & 0 & 0 & 0 & 0 & 0 & \scalemath{0.7}{s} \\
0 & 0 & 0 & 0 & 1 & 0 & 0 & 0 & 0 & 0 & \scalemath{0.7}{c} \\
0 & 0 & 0 & 0 & 0 & 1 & 0 & 0 & 0 & 0 & \scalemath{0.7}{\epsilon} \\
0 & 0 & 0 & 0 & 0 & 0 & 1 & 0 & 0 & 0 & \scalemath{0.7}{k_{-1}} \\
0 & 0 & 0 & 0 & 0 & 0 & 0 & 1 & 0 & 0 & \scalemath{0.7}{k_2} \\
-1 & 0 & 1 & 0 & 0 & 0 & -1 & -1 & 0 & 0 & \scalemath{0.7}{k_1} \\
0 & 0 & 0 & 0 & 0 & 0 & 0 & 0 & 1 & 0 & \scalemath{0.7}{K_m} \\
0 & 0 & 0 & 0 & 0 & 0 & 0 & 0 & 0 & 1 & \scalemath{0.7}{e_0} \\
-1 & 1 & 1 & -1 & -1 & 0 & -1 & -1 & -1 & -1 & \scalemath{0.7}{s_0} \\
\end{block}
\scalemath{0.7}{\frac{1}{s_0 k_1}} & \scalemath{0.7}{s_0} & \scalemath{0.7}{s_0 k_{1} t} & \scalemath{0.7}{\frac{s}{s_{0}}} & \scalemath{0.7}{\frac{c}{s_{0}}} & \scalemath{0.7}{\epsilon} & \scalemath{0.7}{\frac{k_{-1}}{k_{1} s_{0}}} & \scalemath{0.7}{\frac{k_{2}}{k_{1} s_{0}}} & \scalemath{0.7}{\frac{K_m}{s_{0}}} & \scalemath{0.7}{\frac{e_{0}}{s_{0}}} \\
\end{blockarray}
\end{equation}
To find the nondimensionalization of Segel-Slemrod as given in \cite[(6.20)]{murray}, perform the column operations
\begin{equation*}
    V_3' = V_3 + V_{10} - V_6, \quad V_5' = V_5 - V_6, \quad V_9' = -V_9, \quad V_7' = V_7 - V_8
\end{equation*}
with all other columns of $V' \in M_{10}(\mathbb{Z})$ equal to those of $V$.
\begin{equation*}
V' =
\begin{blockarray}{ccccccccccc}
\begin{block}{[cc|cccccccc]c}
0 & 0 & 1 & 0 & 0 & 0 & 0 & 0 & 0 & 0 &  \scalemath{0.7}{t} \\
0 & 0 & 0 & 1 & 0 & 0 & 0 & 0 & 0 & 0 &  \scalemath{0.7}{s} \\
0 & 0 & 0 & 0 & 1 & 0 & 0 & 0 & 0 & 0 &  \scalemath{0.7}{c} \\
0 & 0 & -1 & 0 & -1 & 1 & 0 & 0 & 0 & 0 &  \scalemath{0.7}{\epsilon} \\
0 & 0 & 0 & 0 & 0 & 0 & 1 & 0 & 0 & 0 &  \scalemath{0.7}{k_{-1}} \\
0 & 0 & 0 & 0 & 0 & 0 & -1 & 1 & 0 & 0 &  \scalemath{0.7}{k_2} \\
-1 & 0 & 1 & 0 & 0 & 0 & 0 & -1 & 0 & 0 & \scalemath{0.7}{k_1} \\
0 & 0 & 0 & 0 & 0 & 0 & 0 & 0 & -1 & 0 & \scalemath{0.7}{K_m} \\
0 & 0 & 1 & 0 & 0 & 0 & 0 & 0 & 0 & 1 & \scalemath{0.7}{e_0} \\
-1 & 1 & 0 & -1 & -1 & 0 & 0 & -1 & 1 & -1  & \scalemath{0.7}{s_0} \\
\end{block}
\scalemath{0.7}{\frac{1}{s_0 k_1}} & \scalemath{0.7}{s_0} & \scalemath{0.7}{\frac{e_0 k_{1} t}{\epsilon}} & \scalemath{0.7}{\frac{s}{s_{0}}} & \scalemath{0.7}{\frac{c}{s_{0} \epsilon}} & \scalemath{0.7}{\epsilon} & \scalemath{0.7}{\frac{k_{-1}}{k_2}} & \scalemath{0.7}{\frac{k_{2}}{k_{1} s_{0}}} & \scalemath{0.7}{\frac{s_0}{K_m}} & \scalemath{0.7}{\frac{e_{0}}{s_{0}}} \\
\end{blockarray}
\end{equation*}
The columns of $V'$ give the invariants
\begin{equation} \label{eqn:segel_slemrod_invariants}
    \tau = \frac{e_0k_1t}{\epsilon}, \, 
    (u,v) = \left(\frac{s}{s_0},\frac{c}{s_0 \epsilon}\right), \, 
    (\epsilon,\rho,\theta_1,\sigma,\theta_2) = \left(\epsilon,\frac{k_{-1}}{k_2},\frac{k_2}{k_1 s_0},\frac{s_0}{K_m},\frac{e_0}{s_0}\right).
\end{equation}

\noindent
Together with the two additional%
\footnote{By substituting \eqref{eqn:segel_slemrod_invariants} into equations \eqref{eqn:michaelis_menten_kinetics}, we see (by writing our invariants in terms of the original constants if necessary) that $\theta_1$ and $\theta_2$ can be cancelled in the reduced system, so there is nothing different here from \cite[(6.20)]{murray}.} 
invariants $\theta_1$ and $\theta_2$, these are precisely the dimensionless quantities given in \cite[(6.20)]{murray}. Computing $W=(V')^{-1}$ to find the substitutions
\begin{equation*}
    (t,s,c,\epsilon,k_{-1},k_2,k_1,K_m,e_0,s_0) \longmapsto \left(\frac{\epsilon}{\theta_2}\tau,u,\epsilon v,\epsilon,\rho \theta_1,\theta_1,1,\frac{1}{\sigma},\theta_2,1 \right).
\end{equation*}
Using Theorem~\ref{thm:hubert_labahn_7.1}, we have computed a reduced system
\begin{align*}
    \frac{\theta_2}{\epsilon} \frac{du}{d\tau} &= -\theta_2 u + \epsilon uv + \epsilon \rho \theta_1 v 
    &u(0)&=1
    \\
    \theta_2 \frac{dv}{d\tau} &= \theta_2 u - \epsilon uv - \epsilon \rho \theta_1 v - \theta_1 v
    & v(0)&=0.
\end{align*}

To obtain the familiar reduced system in \cite[(6.21)]{murray} divide by $\theta_2$ in both equations and observe the following relations, writing everything in terms of the original constants in \eqref{eqn:michaelis_menten_kinetics}:
\begin{align*}
\begin{split}
    &\frac{\epsilon}{\theta_2} = \frac{s_0}{s_0 + K_m} = \frac{\sigma}{1+\sigma}, \qquad \frac{\epsilon \rho \theta_1}{\theta_2} = \frac{k_{-1}}{k_1s_0 + k_1 + k_2} = \frac{\rho}{(1+\sigma)(1+\rho)}, \\
    &\frac{\epsilon \theta_1}{\theta_2}(\rho+1) = \frac{1}{s_0 + K_m} \cdot \frac{k_{-1}+k_2}{k_1} = \frac{1}{1+\sigma}.
    \end{split}
\end{align*}
Substituting gives
\begin{align*}
    \frac{du}{d\tau} &= \epsilon \left( -u + \frac{\sigma}{1 + \sigma}uv + \frac{\rho}{(1+\sigma)(1 + \rho)}v \right)   &u(0)&=1  \\ 
    \frac{dv}{d\tau} &= u - \frac{\sigma}{1+\sigma}uv - \frac{v}{1+\sigma}  & v(0)&=0.
\end{align*}

Let us now see how we can obtain the dimensionless model of \cite[(6.23)]{murray} corresponding to the slow timescale $T=\epsilon (1+\rho)k_2t$. %  as given in \cite[(6.22)]{murray}
Observing that $(1+\rho)k_2 = k_1K_m$, we can obtain our new timescale $T=\epsilon (1+\rho)k_2t = \epsilon k_1K_mt$ from the Hermite multiplier $V''$, where $V_3'' = V_3' + 2V_6' + V_9'- V_{10}'$, and $V_i''=V_i'$ otherwise. Compute the substitutions from $W' = (V'')^{-1}$, 
\begin{equation*}
    (t,s,c,\epsilon,k_{-1},k_2,k_1,K_m,e_0,s_0) \longmapsto \left(\frac{\sigma}{\epsilon}T,u,\epsilon v,\epsilon,\rho \theta_1,\theta_1,1,\frac{1}{\sigma},\theta_2, 1 \right)
\end{equation*}
and obtain the familiar reduced system for the slow timescale:
\begin{align*}
\frac{du}{dT} & = -(1+\sigma)u + \sigma uv + \frac{\rho}{1+\rho}v 
&u(0) &= 1
\\
\epsilon \frac{dv}{dT} & = (1+\sigma)u - \sigma uv - v 
&v(0) &= 0. 
\end{align*}

\subsection{New nondimensionalizations} \label{sec:application_new_nondimen}

Let us now see the two reduced systems which our algorithm outputs for Segel-Slemrod analysis \emph{without} performing any operations on our Hermite multiplier $V$. In other words, given the constraints $K_m=\frac{k_{-1}+k_2}{k_1}$ and $\epsilon=\frac{e_0}{s_0 + K_m}$, we can obtain an alternative nondimensionalization by simply using the invariants given in \eqref{eqn:segel_slemrod_invariants}.

Setting $\eta = \frac{k_{-1}}{k_1 s_0}$ and using $\sigma = \frac{s_0}{K_m}$ from our Hermite multiplier $V$ in \eqref{eqn:micheaelis_menten_hermite_multiplier}, we obtain new reduced systems with $u=\frac{s}{s_0}$ and $v=\frac{c}{\epsilon s_0}$ whose fast and slow timescales, respectively $\tau = k_1s_0t$ and $T = \epsilon k_1 s_0 t$, are respectively given by
\begin{align*}
&\frac{du}{d\tau} = \frac{\epsilon}{\sigma} \left(\eta \sigma v + \sigma u v - \sigma u - u\right)  &  \qquad &\frac{du}{dT} = \frac{1}{\sigma} \left(\eta \sigma v + \sigma u v - \sigma u - u\right) \\
&\frac{dv}{d\tau} = - \frac{1}{\sigma} \left(\sigma u v - \sigma u - u + v\right)                   & & \frac{dv}{dT} = - \frac{1}{\epsilon \sigma} \left(\sigma u v - \sigma u - u + v\right) \\
&u\left(0\right) = 1, \quad v(0)=0                                                                  &  &u\left(0\right) = 1, \quad v(0)=0.
\end{align*}

As noted in Section~\ref{sec:michaelis_menten_comparison_with_classical}, it is often assumed that $\epsilon = \frac{e_0}{s_0} \ll 1$. When this is no longer the case, the Segel-Slemrod analysis in Section~\ref{sec:application_segel_slemrod} instead proposes that we should choose $\epsilon = \frac{e_0}{s_0 + K_m}$. How can we motivate this choice of $\epsilon$ within this paper's computational framework?

For a scaling matrix $A$, let us ask for which monomials $x$ is $\epsilon = \frac{e_0}{s_0 + x}$ an invariant of the scaling action defined by $A$. From the possible answers, we may choose the largest such $x$ to produce the smallest possible $\epsilon$. For $\epsilon$ to be $\mathbb{T}_A$ invariant, we need $A$ to act on $x$ as it does on $s_0$, which means that we require $\frac{x}{s_0}$ to be $\mathbb{T}_A$-invariant.\footnotemark

\footnotetext{Since $\epsilon = \frac{e_0}{s_0 + x} = \frac{\frac{e_0}{s_0}}{1 + \frac{s_0}{x}}$, and because \eqref{eqn:mm_column_hermite_multiplier} gives that $\frac{e_0}{s_0}$ is a scaling invariant, $\epsilon$ is invariant precisely when $\frac{x}{s_0}$ is invariant.}

In light of Remark~\ref{remark:hubert_labahn6.5_rational_invariants_by_substitution}, the last seven columns of $V$ in  \eqref{eqn:mm_column_hermite_multiplier} generate all $\mathbb{T}_A$-invariants,  so
\begin{equation} \label{eqn:x_over_s0}
    \frac{x}{s_0} = \left( k_{1} s_{0} t \right)^{n_1}
\left( \frac{s}{s_{0}} \right)^{n_2}
\left( \frac{c}{e_{0}} \right)^{n_3}
\left( \frac{K_m}{s_0} \right)^{n_4}
\left( \frac{k_{-1}}{k_1 s_0} \right)^{n_5}
\left( \frac{k_2}{k_1 s_0} \right)^{n_6}
\left( \frac{e_0}{s_0} \right)^{n_7}
\end{equation}
\noindent
for some integers $n_1,\ldots,n_7 \in \mathbb{Z}$. Choosing $n_4=1$ and $n_i=0$ otherwise, we recover the choice of $\epsilon$ used in Segel-Slemrod analysis. But what about finding new $\epsilon$ with which to nondimensionalize?

Remembering that we want $x$ to be a function only of the constant parameters, \eqref{eqn:x_over_s0} suggests that the simplest candidates for $x$ are $e_0,\frac{k_{-1}}{k_1}$ and $\frac{k_2}{k_1}$. The first case results in $\epsilon = \frac{e_0}{s_0 + e_0}$, and this is only negligible if $e_0$ \(\ll\)\ $s_0$. The latter two cases result in a choice of $\epsilon$ which is at least as big as $\frac{e_0}{s_0 + K_m}$. We could additionally try other choices of $x$ which may be useful for certain models.

%--------------------------------------------------%

\subsection{Extending invariants in four-variable Michaelis-Menten}  \label{sec:michaelis_menten_four_variables}

We return to the Michaelis-Menten model, now considered in its four-variable form \cite[(6.3)]{murray}:
\begin{align*}
    \frac{de}{dt} &= -k_1es+k_{-1}c+k_2c &
    \frac{ds}{dt} &= -k_1es + k_{-1}c \\
    \frac{dc}{dt} &= k_1es - k_{-1}c - k_2c &
    \frac{dp}{dt} &= k_2c
\end{align*}

Choosing variable order $(t,c,e,p,s,k_1,k_2,k_{-1})$, our scaling matrix is
\begin{equation*}
    A = \begin{bmatrix} 1 & 0 & 0 & 0 & 0 & -1 & -1 & -1 \\ 0 & 1 & 1 & 1 & 1 & -1 & 0 & 0 \end{bmatrix}.
\end{equation*}
The first row corresponds to a scaling of time, and the second to a scaling of concentrations. The Hermite multiplier and resulting invariants are:
\begin{equation*}
    V = \begin{blockarray}{cccccccc}
    \begin{block}{[cccccccc]}
    0 & 0 & 1 & 0 & 0 & 0 & 0 & 0 \\
    0 & 0 & 0 & 1 & 0 & 0 & 0 & 0 \\
    0 & 0 & 0 & 0 & 1 & 0 & 0 & 0 \\
    0 & 0 & 0 & 0 & 0 & 1 & 0 & 0 \\
    0 & 0 & 0 & 0 & 0 & 0 & 1 & 0 \\
    0 & -1 & 0 & 1 & 1 & 1 & 1 & 0 \\
    -1 & 1 & 1 & -1 & -1 & -1 & -1 & -1 \\
    \end{block}
    & & \scalemath{0.7}{k_{-1}t} & \scalemath{0.7}{\frac{k_1c}{k_{-1}}} & \scalemath{0.7}{\frac{k_1e}{k_{-1}}} & \scalemath{0.7}{\frac{k_1p}{k_{-1}}} & \scalemath{0.7}{\frac{k_1s}{k_{-1}}} & \scalemath{0.7}{\frac{k_2}{k_{-1}}}
    \end{blockarray}
\end{equation*}
Suppose we prefer to normalize by $k_2$ rather than $k_{-1}$. Simply place $k_2$ after $k_{-1}$ in the variable order, and  re-compute $A$ and $V$. This yields these new invariants:
\begin{equation*}
    k_2t, \quad \frac{k_1c}{k_2}, \quad \frac{k_1e}{k_2}, \quad \frac{k_1p}{k_2}, \quad \frac{k_1s}{k_2}, \quad \frac{k_{-1}}{k_2}.
\end{equation*}
But suppose instead we want a system with chosen invariants $\frac{k_1}{k_2}c$ and $\frac{k_1}{k_{-1}}p$. Then we may use the algorithm presented in Section~\ref{sec:extension_algorithm}. We keep our original variable order and consider the matrix which represents the invariants $\frac{k_1}{k_2}c$ and $\frac{k_1}{k_{-1}}p$: 
\begin{equation*}
P = 
\begin{bmatrix} 0 & 1 & 0 & 0 & 0 & 1 & -1 & 0 \\ 0 & 0 & 0 & 1 & 0 & 1 & 0 & -1 \end{bmatrix}^\intercal.
\end{equation*}
One can readily confirm that $AP=0$, so these invariants are dimensionally compatible.
Next, compute the Smith normal form decomposition of $W_{\mathfrak{b}}P$ and verify that $U W_{\mathfrak{b}}P U' = S$ has the required form:
\begin{equation*}
\left[\begin{matrix}0 & 1 & 0 & 0 & 0 & 0\\0 & 0 & 0 & 1 & 0 & 0\\0 & 1 & 0 & 0 & 0 & 1\\0 & 0 & 0 & 0 & 1 & 0\\0 & 0 & 1 & 0 & 0 & 0\\1 & 0 & 0 & 0 & 0 & 0\end{matrix}\right] % U
\left[\begin{matrix}1 & 0 & 0 & 0 & 0 & 0 & 0 & 0\\0 & 1 & 0 & 0 & 0 & 0 & 0 & 0\\0 & 0 & 1 & 0 & 0 & 0 & 0 & 0\\0 & 0 & 0 & 1 & 0 & 0 & 0 & 0\\0 & 0 & 0 & 0 & 1 & 0 & 0 & 0\\0 & 0 & 0 & 0 & 0 & 0 & 1 & 0\end{matrix}\right] % Wb
\cdot P \cdot
\left[\begin{matrix}1 & 0\\0 & 1\end{matrix}\right]
=
\left[\begin{matrix}1 & 0\\0 & 1\\0 & 0\\0 & 0\\0 & 0\\0 & 0\end{matrix}\right]
=
S
\end{equation*}
Compute the unimodular matrix $C$ and desired invariants given by $\tilde{V}_{\mathfrak{b}}$:
\begin{equation*}
    U^{-1} = \left[\begin{matrix}0 & 0 & 0 & 0 & 0 & 1\\1 & 0 & 0 & 0 & 0 & 0\\0 & 0 & 0 & 0 & 1 & 0\\0 & 1 & 0 & 0 & 0 & 0\\0 & 0 & 0 & 1 & 0 & 0\\-1 & 0 & 1 & 0 & 0 & 0\end{matrix}\right], 
    \,
    W_{\mathfrak{b}}P = \left[\begin{matrix}0 & 0\\1 & 0\\0 & 0\\0 & 1\\0 & 0\\-1 & 0\end{matrix}\right] 
    \,
    \Rightarrow 
    \,
    C = \left[\begin{matrix}0 & 0 & 0 & 0 & 0 & 1\\1 & 0 & 0 & 0 & 0 & 0\\0 & 0 & 0 & 0 & 1 & 0\\0 & 1 & 0 & 0 & 0 & 0\\0 & 0 & 0 & 1 & 0 & 0\\-1 & 0 & 1 & 0 & 0 & 0\end{matrix}\right]
\end{equation*}
\begin{equation*}
    \tilde{V}_{\mathfrak{b}} = V_{\mathfrak{b}}C = \left[\begin{matrix}0 & 0 & 0 & 0 & 0 & 1\\1 & 0 & 0 & 0 & 0 & 0\\0 & 0 & 0 & 0 & 1 & 0\\0 & 1 & 0 & 0 & 0 & 0\\0 & 0 & 0 & 1 & 0 & 0\\1 & 1 & 0 & 1 & 1 & 0\\-1 & 0 & 1 & 0 & 0 & 0\\0 & -1 & -1 & -1 & -1 & 1\end{matrix}\right]
\rightsquigarrow
        \left[\begin{matrix}1 & 0 & 0 & 0 & 0 & 0\\0 & 1 & 0 & 0 & 0 & 0\\0 & 0 & 1 & 0 & 0 & 0\\0 & 0 & 0 & 1 & 0 & 0\\0 & 0 & 0 & 0 & 1 & 0\\0 & 1 & 1 & 1 & 1 & 0\\0 & -1 & 0 & 0 & 0 & 1\\1 & 0 & -1 & -1 & -1 & -1\end{matrix}\right],
\end{equation*}
after performing admissible column operations on $V_{\mathfrak{b}}$.  The invariants are
\begin{align*}
\begin{split}
\begin{bmatrix}
\tilde{t} & \tilde{c} & \tilde{e} & \tilde{p} & \tilde{s} & \alpha
\end{bmatrix}
 &=
\begin{bmatrix}
k_{-1} t & \frac{c k_{1}}{k_{2}} & \frac{e k_{1}}{k_{-1}} &  \frac{p k_{1}}{k_{-1}} & \frac{s k_{1}}{k_{-1}} & \frac{k_{2}}{k_{-1}}
\end{bmatrix}.
\end{split}
\end{align*}
Substituting yields a system with just one constant parameter $\alpha$: 
\begin{align}
\frac{d\tilde{c}}{d\tilde{t}} &= - \alpha^{2} \tilde{c} - \alpha \tilde{c} + \tilde{e} \tilde{s} &
\frac{d\tilde{e}}{d\tilde{t}} &= \alpha^{2} \tilde{c} + \alpha \tilde{c} - \tilde{e} \tilde{s} \nonumber \\
\frac{d\tilde{p}}{d\tilde{t}} &= \alpha^{2} \tilde{c} &
\frac{d\tilde{s}}{d\tilde{t}} &= \alpha \tilde{c} - \tilde{e} \tilde{s}.
\label{eqn:mm4var_reduced_one_parameter}
\end{align}
\noindent
System \eqref{eqn:mm4var_reduced_one_parameter} is reduced as far as possible, computed automatically by software.

\subsection{Cell cycle control network}  \label{sec:cell_cycle_control}

Consider this cell cycle control network from \cite[Figure 2]{sible-tyson}\footnote{In \cite{sible-tyson}, the variable $\textrm{[APC*]}$ has a typo and should just be $\textrm{[APC]}$. }:
\begin{align*}
\frac{d}{dt}\textrm{[Cyclin]} &= k_1 - k_2 \textrm{[Cyclin]} - k_3 \textrm{[Cyclin]} \textrm{[Cdk]} \\
\frac{d}{dt}\textrm{[MPF]} &= k_3 \textrm{[Cyclin]} \textrm{[Cdk]} - k_2 \textrm{[MPF]} - k_{\textrm{wee}} \textrm{[MPF]} + k_{25} \textrm{[preMPF]} \\
\frac{d}{dt}\textrm{[preMPF]} &= -k_2 \textrm{[preMPF]} + k_{\textrm{wee}} \textrm{[MPF]} - k_{25} \textrm{[preMPF]} \\
\frac{d}{dt}\textrm{[Cdc25P]} &= \frac{k_a \textrm{[MPF]}\left( \textrm{[total Cdc25]} - \textrm{[Cdc25P]} \right)}{K_a + \textrm{[total Cdc25]} - \textrm{[Cdc25P]}} - \frac{k_b \textrm{[PPase]} \textrm{[Cdc25P]}}{K_b + \textrm{[Cdc25P]}} \\
\frac{d}{dt}\textrm{[Wee1P]} &= \frac{k_e \textrm{[MPF]}\left( \textrm{[total Wee1]} - \textrm{[Wee1P]} \right)}{K_e + \textrm{[total Wee1]} - \textrm{[Wee1P]}} - \frac{k_f \textrm{[PPase]} \textrm{[Wee1P]}}{K_f + \textrm{[Wee1P]}} \\
\frac{d}{dt}\textrm{[IEP]} &= \frac{k_g \textrm{[MPF]}\left( \textrm{[total IE]} - \textrm{[IEP]} \right)}{K_g + \textrm{[total IE]} - \textrm{[IEP]}} - \frac{k_h \textrm{[PPase]} \textrm{[IEP]}}{K_h + \textrm{[IEP]}} \\
\frac{d}{dt}\textrm{[APC]} &= \frac{k_c \textrm{[MPF]}\left( \textrm{[total APC]} - \textrm{[APC]} \right)}{K_c + \textrm{[total APC]} - \textrm{[APC]}} - \frac{k_d \textrm{[PPase]} \textrm{[APC]}}{K_d + \textrm{[APC]}} \\
\end{align*}
with the constraints
\begin{align*}
\textrm{[Cdk]} &= \textrm{[total Cdk]} - \textrm{[MPF]} - \textrm{[preMPF]} \\
k_{25} &= V_{25}' \left( \textrm{[total Cdc25]} - \textrm{[Cdc25P]} \right) + V_{25}'' \textrm{[Cdc25P]} \\
k_{\textrm{wee}} &= V_{\textrm{wee}}' \textrm{[Wee1P]} + V_{\textrm{wee}}'' \left( \textrm{[total Wee1]} - \textrm{[Wee1P]} \right) \\
k_2 &= V_2' \left( \textrm{[total APC]} - \textrm{[APC]} \right) + V_2'' \textrm{[APC]}.
\end{align*}

Because the constraints involve non-constant variables, we substitute them into the equations before reduction.  We chose a variable order ending with total quantities so the algorithm will more likely scale by them,
\begin{gather*}
(
t,
% the dependent variables, in order of equations to help things stay easier to read
[\textrm{Cyclin}],
[\textrm{MPF}],
[\textrm{preMPF}],
[\textrm{Cdc25P}],
[\textrm{Wee1P}],
[\textrm{IEP}],
[\textrm{APC}], \\
% start the parameters
[\textrm{PPase}], 
K_a,
K_b,
K_c,
K_d,
K_e,
K_f,
K_g,
K_h, 
k_1,
k_3,
k_a,
k_b,
k_c,
k_d,
k_e,
k_f,
k_g,
k_h, \\
V_{25}',
V_{25}'',
V_{\textrm{wee}}',
V_{\textrm{wee}}'',
V_2',
V_2'', \\ % grouping to keep similar together, visually
[\textrm{total Cdc25}], % the totals are last so that the software will likely normalize with respect to them
[\textrm{total IE}],
[\textrm{total APC}],
[\textrm{total Cdk}],
[\textrm{total Wee1}]
).
\end{gather*}

The software \verb|desr| yields a $7 \times 38$ scaling matrix, and produces a reduced system with $7$ fewer symbols. The rational invariants are 

\begin{gather*}
V_2'' t \textrm{[total APC]},
\frac{\textrm{[Cyclin]}}{\textrm{[total Cdk]}},
\frac{\textrm{[MPF]}}{\textrm{[total Cdk]}},
\frac{\textrm{[preMPF]}}{\textrm{[total Cdk]}},
\frac{\textrm{[Cdc25P]}}{\textrm{[total Cdc25]}},
\frac{\textrm{[Wee]}1P}{\textrm{[total Wee1]}},\\
\frac{\textrm{[IEP]}}{\textrm{[total IE]}}, 
\frac{\textrm{[APC]}}{\textrm{[total APC]}},
\frac{\textrm{[PPase]} k_h}{(V_2'' \textrm{[total APC]} \textrm{[total IE]})},
\frac{K_a}{\textrm{[total Cdc25]}},
\frac{K_b}{\textrm{[total Cdc25]}},\\
\frac{K_c}{\textrm{[total APC]}},
\frac{K_d}{\textrm{[total APC]}},
\frac{K_e}{\textrm{[total Wee1]}},
\frac{K_f}{\textrm{[total Wee1]}},
\frac{K_g}{\textrm{[total IE]}},
\frac{K_h}{\textrm{[total IE]}},\\
\frac{k_1}{V_2'' \textrm{[total APC]} \textrm{[total Cdk]}},
\frac{k_3 \textrm{[total Cdk]}}{(V_2'' \textrm{[total APC]})},
\frac{k_a \textrm{[total Cdk]}}{V_2'' \textrm{[total APC]} \textrm{[total Cdc25]}},\\
\frac{k_b \textrm{[total IE]}}{(k_h \textrm{[total Cdc25]})},
\frac{k_c \textrm{[total Cdk]}}{V_2'' \textrm{[total APC]} 2},
\frac{k_d \textrm{[total IE]}}{(k_h \textrm{[total APC]})},
\frac{k_e \textrm{[total Cdk]}}{V_2'' \textrm{[total APC]} \textrm{[total Wee1]}},\\
\frac{k_f \textrm{[total IE]}}{(k_h \textrm{[total Wee1]})},
\frac{k_g \textrm{[total Cdk]}}{V_2'' \textrm{[total APC]} \textrm{[total IE]}},
\frac{V_{25}' \textrm{[total Cdc25]}}{V_2'' \textrm{[total APC]}},
\frac{V_{25}'' \textrm{[total Cdc25]}}{V_2'' \textrm{[total APC]}},\\
\frac{V_{\textrm{wee}}' \textrm{[total Wee1]}}{V_2'' \textrm{[total APC]}},
\frac{V_{\textrm{wee}}'' \textrm{[total Wee1]}}{V_2'' \textrm{[total APC]}},
V_2'/V_2''.
\end{gather*}
After using \verb|sympy.collect| on all variables in order, the reduced system is
\begin{align*}
\frac{d\nu_{1}}{d\tau} &= \kappa_{10} + \nu_{1} \left(\kappa_{11} \left(\nu_{2} + \nu_{3} - 1\right) - \kappa_{23} + \nu_{7} \left(\kappa_{23} - 1\right)\right) \\
\frac{d\nu_{2}}{d\tau} &= \kappa_{11} \nu_{1} \left(- \nu_{2} - \nu_{3} + 1\right) + \nu_{2} \left(- \kappa_{22} - \kappa_{23} + \nu_{5} \left(- \kappa_{21} + \kappa_{22}\right) + \nu_{7} \left(\kappa_{23} - 1\right)\right) + \\ & \qquad \nu_{3} \left(\kappa_{19} + \nu_{4} \left(- \kappa_{19} + \kappa_{20}\right)\right) \\
\frac{d\nu_{3}}{d\tau} &= \nu_{2} \left(\kappa_{22} + \nu_{5} \left(\kappa_{21} - \kappa_{22}\right)\right) + \nu_{3} \left(- \kappa_{19} - \kappa_{23} + \nu_{4} \left(\kappa_{19} - \kappa_{20}\right) + \nu_{7} \left(\kappa_{23} - 1\right)\right) \\
\frac{d\nu_{4}}{d\tau} &= - \frac{\kappa_{1} \kappa_{13} \nu_{4}}{\kappa_{3} + \nu_{4}} + \kappa_{12} \nu_{2} \left(- \frac{\nu_{4}}{\kappa_{2} - \nu_{4} + 1} + \frac{1}{\kappa_{2} - \nu_{4} + 1}\right) \\
\frac{d\nu_{5}}{d\tau} &= - \frac{\kappa_{1} \kappa_{17} \nu_{5}}{\kappa_{7} + \nu_{5}} + \kappa_{16} \nu_{2} \left(- \frac{\nu_{5}}{\kappa_{6} - \nu_{5} + 1} + \frac{1}{\kappa_{6} - \nu_{5} + 1}\right) \\
\frac{d\nu_{6}}{d\tau} &= - \frac{\kappa_{1} \nu_{6}}{\kappa_{9} + \nu_{6}} + \kappa_{18} \nu_{2} \left(- \frac{\nu_{6}}{\kappa_{8} - \nu_{6} + 1} + \frac{1}{\kappa_{8} - \nu_{6} + 1}\right) \\
\frac{d\nu_{7}}{d\tau} &= - \frac{\kappa_{1} \kappa_{15} \nu_{7}}{\kappa_{5} + \nu_{7}} + \kappa_{14} \nu_{2} \left(- \frac{\nu_{7}}{\kappa_{4} - \nu_{7} + 1} + \frac{1}{\kappa_{4} - \nu_{7} + 1}\right).
\end{align*}
where each of $V_2'', [\textrm{total Cdc25}], [\textrm{total IE}], [\textrm{total APC}], [\textrm{total Cdk}], [\textrm{total Wee1}]$ have been algorithmically normalized to $1$.

\subsection{A vaccine injection model} \label{sec:compartment model}

Consider the vaccine injection model from \cite{WeissburgEtal}, as presented in \cite[Example 13.6]{DiStefano}, consisting of a linear two-compartment model with leaks in each compartment, input in compartment 1, and output in compartment 2. The corresponding ODE system is
\begin{align*}
    \frac{dx_1}{dt} =-(a_{01}+a_{21})x_1+u_1 && \frac{dx_2}{dt} =a_{21}-a_{02}x_2 && y_2=x_2.
\end{align*}
In this system, $y_2 = x_2$ appears because measurement happens via $y_2$. The software \verb|desr| yields a 2-dimensional scaling symmetry.

Now if we make the change of variables \(b_{11}=a_{01}+a_{21}\), with everything else remaining the same, we get the system:
\begin{align*}
     \frac{dx_1}{dt} =-b_{11}x_1+u_1 & & \frac{dx_2}{dt} =a_{21}-a_{02}x_2 & & y_2=x_2.
\end{align*}
 For this new system, the software \verb|desr| yields a 3-dimensional scaling symmetry. 
 
 In light of Theorem \ref{thm:ChangeOfVariables}, what has happened here is that when we applied the software to the new system blindly, we did not take into account that the new constant \(b_{11}\) needs to have the same units as \(a_{21}\) if we want to ensure our change of variables is dimensionally consistent. If we add this as a constraint for the new system, then the software \verb|desr| now yields a 2-dimensional scaling symmetry as in the original system.

\section{Conclusion}
\label{sec:conclusion}

This paper presents a practical algebraic framework for nondimensionalizing systems of ordinary differential equations. Building on and extending the Hubert–Labahn scaling-symmetry algorithm, we show how to compute maximal scaling symmetry reduction (Theorem~\ref{thm:hubert_labahn_7.1}) while enforcing dimensional consistency (Theorem~\ref{thm:ChangeOfVariables}). We further incorporate considerations such as initial conditions, conserved quantities and known parameter invariants. Implemented in the Python package \texttt{desr}, this framework provides mathematical scientists with a scalable, principled tool for automated nondimensionalization. The algorithm and implementation does have limitations: it finds only scaling symmetries, not slow or fast time scales arising from sums of variables or other expressions from non-scaling actions. 

Through detailed applications in mathematical biology, we demonstrate how to recover classical nondimensionalizations and discover new ones beyond the reach of the Buckingham-$\pi$ theorem.  Looking ahead, there are possible extensions and directions, such as spatial model analysis \cite{parres2025contextual,plank2025random}. Concretely, these methods open the door to integrating nondimensionalization into symbolic-numeric workflows for ODE model reduction \cite{feliu2022quasi}, parameter inference \cite{liu2024approximate,meshkat2025structural}, and multiscale analysis \cite{astuto2023multiscale}, and suggest future links with algebraic geometry, symmetry-based machine learning \cite{desai2022symmetry}, and automated discovery pipelines in mathematical models \cite{liu2022hierarchical}.

\section*{Acknowledgments}

The authors thank Siddharth Unnithan Kumar for preliminary work and discussions that led to this manuscript.

\bibliographystyle{siamplain}
\bibliography{references}

\end{document}